\documentclass[12pt]{article}
\usepackage{amsmath}
\usepackage{graphicx}
\usepackage{natbib}
\usepackage{url} 

\newcommand{\blind}{0}

\usepackage{mathtools,amsthm,amssymb}
\usepackage{xcolor}
\usepackage{booktabs}

\newcommand{\dif}{\mathrm{d}}
\newcommand{\tr}{\operatorname{tr}}

\newtheorem{prop}{Proposition}[section]
\newtheorem{cor}[prop]{Corollary}
\newtheorem{defi}[prop]{Definition}

\newtheorem{exam}[prop]{Example}

\newtheorem{lem}[prop]{Lemma}

\newtheorem{theo}[prop]{Theorem}

\newtheorem{rem}{Remark}

\newcommand{\revision}[1]{#1}

\newcommand{\EE}{\mathbb{E}}

\newcommand{\RR}{\mathbb{R}}

\def \RR{\mathbb{R}}

\def \EE{\mathbb{E}}

\def \U{\mathcal{U}}

\begin{document}

\def\spacingset#1{\renewcommand{\baselinestretch}%
{#1}\small\normalsize} \spacingset{1}


\if0\blind
{
  \title{\bf MCMC Algorithms for Posteriors \\ on Matrix Spaces}
  \author{Alexandros Beskos\thanks{
    Alexandros Beskos acknowledges support from a Leverhulme Trust Prize.}\hspace{.2cm}\\
    Department of Statistical Science, University College London\\
    and \\
    Kengo Kamatani\thanks{
    Kengo Kamatani acknowledges support from JSPS KAKENHI Grant Numbers 16K00046, 20H04149 and JST CREST Grant Number JPMJCR14D7.
    }\hspace{.2cm}
    \\
    Department of Engineering Science, Osaka University}
  \maketitle
} \fi

\if1\blind
{
  \bigskip
  \bigskip
  \bigskip
  \begin{center}
    {\LARGE\bf Title}
\end{center}
  \medskip
} \fi

\bigskip
\begin{abstract}
We study Markov chain Monte Carlo (MCMC) algorithms for target distributions defined 
on matrix spaces.
Such an important sampling problem has yet to be analytically explored. We carry out a major step in covering this gap by developing the proper theoretical framework that allows for the identification of ergodicity properties of typical MCMC algorithms, relevant in such a context. Beyond the standard Random-Walk Metropolis (RWM) and preconditioned Crank--Nicolson (pCN), a contribution of this paper in the development of a novel algorithm, termed the `Mixed' pCN (MpCN).
RWM and pCN are shown \emph{not} to be geometrically ergodic for an important class of matrix distributions with heavy tails. In contrast, MpCN is robust across targets with different tail behaviour and has very good empirical performance within the class of
heavy-tailed distributions. Geometric ergodicity for MpCN is not fully proven in this work, as some remaining drift conditions are quite challenging to obtain owing to the complexity of the state space. 
We do, however, make a lot of progress towards a proof, and show  in detail the last steps left for future work.
 We illustrate the computational performance of the various algorithms through  numerical applications,
 including calibration on real data of a challenging model arising in financial statistics.
\end{abstract}

\noindent%
{\it Keywords:}  Preconditioned Crank--Nicolson; Drift Condition; Matrix-Valued Stochastic Differential Equation.
\vfill

\newpage
\spacingset{1.5} 
\section{Introduction}
\label{sec:intro}

Statistical models with parameters defined on matrix spaces arise naturally in many applications, with maybe most typical the case of covariance matrices. Related measures have thus been developed, with most prominent the Matrix-Normal, Wishart and Inverse-Wishart distributions
\citep{barn:00}. Numerous extensions have appeared, see e.g.\@ \cite{mall:08, huan:13} for a scaled and an hierarchical Inverse-Wishart, \cite{barn:00} for a strategy that extracts the correlation matrix, \cite{rove:02} and \cite{dobr:2011} for the Hyper-Inverse-Wishart and G-Wishart distributions, respectively.

Beyond standard conjugate settings -- e.g.~Inverse-Wishart prior for the covariance of Gaussian observations -- more involved hierarchical models have generated a need for developing a suite of accompanying MCMC methods. 
This work reviews standard algorithms and introduces a novel one (MpCN), motivated by an MCMC method on vector-spaces used in \cite{kama:17}. The paper invokes a theoretical framework that permit the analysis of the ergodicity properties of some of the presented algorithms or -- in the case of MpCN -- makes a lot of progress towards a (quite challenging) proof, and describes the last remaining steps left for future research. 
%
We show that RWM and pCN do not work well even for an Inverse-Wishart target as they are not geometrically ergodic  for heavy-tailed distributions. In contrast, MpCN has much better empirical performance on such targets.
%
%
Our main contributions are summarised as follows.
%

(i) We develop a new MCMC method -- MpCN --, 
and provide a motivation for its underpinnings.
MpCN is characterised by better empirical performance against RWM or pCN, in a number of numerical studies.
%

(ii) We prove that targets on the space of positive definite matrices can be `upcasted' onto corresponding laws on unrestricted matrices, the latter space permitting a direct path for the development of MCMC methodology. 
%
%

(iii)
We prove that  RWM, pCN are not geometrically ergodic for a wide class of matrix-valued targets. We make a lot of progress into demonstrating geometric ergodicity for MpCN, and highlight the remaining steps for the completed proof.

(iv)
 We run MpCN on a challenging hierarchical model providing a matrix-extension of the  influential scalar Stochastic Volatility (SV) dynamic jump-model by \cite{barn:01}. We stress that the paper focuses on MCMC methods with `blind' proposals. This is in agreement with the selected SV application -- and, more generally, modern pseudo-marginal methods \citep{MR2502648} for complex models -- where guided proposals or Gibbs sampler schemes are typically cumbersome and impractical. 

\revision{We note that there are at least two main difficulties in constructing MCMC on a matrix space. First, matrix calculations such as multiplication, inversion and eigenvalue decomposition can be expensive. So, attempts to use gradient-based methods can lead to prohibitively high computational costs. In this work, consideration of derivatives is completely avoided. Second, involved matrices may contain specific structure, thus careful selection of the proposal kernel is required for such structure to be preserved.  For example, we mainly consider symmetric positive definite matrix spaces representing a cone in the space of diagonal matrices.  Thus, the na\"ive random-walk Metropolis cannot be used as it will not preserve positive definiteness.  }

The paper develops as follows.
Section \ref{sec:pm} introduces relevant measures on matrix spaces.
Section \ref{sec:MCMC} presents MCMC algorithms on such spaces.
Section~\ref{sec:ergodicity} develops ergodicity results for the RWM and pCN algorithms. Section~\ref{sec:mpcn-ergodicity} investigates ergodicity properties of MpCN.
Section \ref{sec:results} shows a collection of numerical results. 
Section \ref{sec:conclude} provides conclusions and points at future work.

\textbf{Notation:}
$\mathbb{R}_+:=(0,\infty)$. 
We write $X=_d Y$ if variables $X$, $Y$ have the same law. 
$M(p,q)$ is the set of $p\times q$ (real-valued) matrices,
$\mathrm{GL}(p)$ the group of $p\times p$ invertible matrices, $\mathrm{Sym}(p)$  the set of $p\times p$ symmetric matrices, $P^+(p)$ the set of $p\times p$ symmetric, positive-definite matrices, 
$\mathcal{O}(p)$ the set of orthogonal matrices, 
$\mathcal{O}(p,q)$, $p\ge q$, the space of $p\times q$ matrices $\mathsf{U}$ such that $\mathsf{U}^{\top}\mathsf{U}=I_q$.
%
%
We use the notation $A=(A_{ij})$ 
to indicate individual matrix elements.
Let
$\|A\|_F$ be the  Frobenius norm, with inner product
$\langle A, B\rangle_F=\tr(A^{\top}B)$, where $\tr(\cdot)$ is the trace of a matrix. 
$A^{\top}$ is the transpose of $A$, and $\det(A)$ its determinant.   
The derivative of $h:M(p,q)\rightarrow\mathbb{R}$
is the $p\times q$ matrix 
$
D h(A):=({\partial h(A)}/{\partial A_{ij}}). 
$
For $A\in P^{+}(p)$, we denote by $A^{1/2}$ the matrix $B\in  P^+(p)$ such that $BB=A$.

\section{Measures on Matrix Spaces}
\label{sec:pm}
%
%
\subsection{Reference Measures}
\label{sec:reference}
%
For $\U \in P^+(p)$, $p\ge q$, we define $\nu_{\,\U }$ on $M(p,q)$, 
\begin{gather*}
\nu_{\,\U }(\dif X):=
\frac{\Gamma_q(p/2)}{\pi^{pq/2}}
\big(\frac{(\det \U )^{-q/2}}{\det (X^{\top}\U ^{-1}X)^{p/2}}\big)\operatorname{Leb}(\dif X),\\
\operatorname{Leb}(\dif X)=\prod_{i=1}^p\prod_{j=1}^q\dif X_{ij}, 
\end{gather*}
$\Gamma_p(\cdot)$ is the multivariate Gamma function.
We denote $\nu_{\,\U }$ simply as $\nu$  when $\U =I_p$.
$\mathrm{Leb}$ denotes the Lebesgue measure, on a space of dimensions implied by the context.
The Lebesgue measure and $\nu_{\,\U }$ will be used as reference measures on $M(p,q)$. Note that $\nu_{\,\U }$ is invariant on $M(p,q)$ under right multiplication, $X\mapsto Xa$, $a\in \mathrm{GL}(q)$. If $p=q$ then $\nu\equiv \nu_{\,\U }$ is a unimodular Haar measure on the locally compact topological group $\mathrm{GL}(q)$ and 
$
\nu\equiv\nu_{-1},
$
where $\nu_{-1}(A):= \nu(A^{-1})$
for any Borel set $A\subseteq \mathrm{GL}(q)$;
see e.g.~Sections 5, 7 of \cite{MR770934}, Section 60 of \cite{MR0033869}.

On $P^+(q)$, we define the measure, 
\begin{equation}
\label{eq:mmu}
\mu(\dif S):=(\det S)^{-(q+1)/2}\prod_{1\le i\le j\le q} \dif S_{ij}.
\end{equation}
We note that $\mu$ is invariant under the transform $S\mapsto aSa^{\top}$, $a\in\mathrm{GL}(q)$;
see Section 5 of \cite{MR770934}.
Also, as in Problem 7.10.6 of \cite{MR770934}, 
$
\mu\equiv \mu_{-1}
$.
%

For the compact topological space $\mathcal{O}(p,q)$, there is a uniform probability distribution that will be denoted $\dif\mathsf{U}$ in this work. 
Note that the uniform distribution on  $\mathcal{O}(p,q)$, for $p\ge q$,
is the marginal on $\mathcal{O}(p,q)$ of the uniform distribution on $\mathcal{O}(p)$, see Theorem 3.3.1 of \cite{MR1960435}.  
\revision{
The reference measures satisfy,
\begin{equation}
\int_{M(p,q)}f(x)\nu(\dif x)=\int_{M(p,q)}f(\mathcal{U}^{-1/2}x)\nu_{\mathcal{U}}(\dif x)=\int_{P^+(q)}\mu(\dif S)~\int_{\mathcal{O}(p,q)}f(\mathsf{U} S^{1/2})\dif \mathsf{U},
\label{eq:farrell}
\end{equation}
for $f:M(p,q)\rightarrow\mathbb{R}$. 
}


\subsection{Probability Measures}
\label{subsec:pm}

\begin{exam}
	For parameters $M\in M(p,q)$, $\Sigma\in P^+(p)$, $\mathrm{T}\in P^+(q)$, the Matrix-Normal distribution $N_{p,q}(M,\Sigma,\mathrm{T})$ has density
	with respect to $\operatorname{Leb}$, \vspace{0.2cm}
	\begin{align*}
	\phi_{p,q}(X;M,\Sigma,\mathrm{T})=
	\frac{\exp\big(-\tr\,[\,\mathrm{T}^{-1}(X-M)^{\top}\Sigma^{-1}(X-M)\,]/2\,\big)}{(2\pi)^{pq/2}(\det \Sigma)^{q/2}(\det \mathrm{T})^{p/2}}. 
	\end{align*}
\vspace{-0.7cm}	
\end{exam}	
%



\begin{exam}
Let $r\in\RR$, $r>q-1$,  and $T\in P^+(q)$. 
The Wishart distribution $W_q(r,T)$ is a probability measure on $P^+(q)$ with density 
with respect to $\mu$, \vspace{0.2cm}
\begin{equation}
\pi(S)=\frac{\det(S)^{r/2}\exp(-\tr\,[\,T^{-1}S\,]/2\,)}{2^{rq/2}(\det T)^{r/2}\Gamma_q(r/2)}. 
\label{eq:WD}
\end{equation}
%
\end{exam}


\begin{exam}
Let $r\in\RR$, $r>q-1$, and $T\in P^+(q)$. 
The Inverse-Wishart distribution $W_q^{-1}(r,T)$ is a probability measure on $P^+(q)$ with density with respect to $\mu$, \vspace{0.2cm}	
$$
\pi(S)=\frac{(\det T)^{r/2}\exp(-\tr\,[\,TS^{-1}\,]/2\,)}{2^{rq/2}\det(S)^{r/2}\Gamma_q(r/2)}. \vspace{0.2cm}	
$$
\end{exam}
\begin{rem}
\label{rem:connect}
We summarize some relevant distribution properties. 
\vspace{-0.2cm}
\begin{itemize}
    \item[(i)] If $X\sim N_{p,q}(M,\Sigma,\mathrm{T})$ then,  \vspace{-0.1cm}
\begin{equation*}
 \mathbb{E}\,[\,(X-M)(X-M)^{\top}\,] = \Sigma\tr(\mathrm{T}),\quad  
 \mathbb{E}\,[\,(X-M)^{\top}(X-M)\,] = \mathrm{T}\tr(\Sigma). 
\end{equation*} 
Also, for $A\in M(r,p)$, $B\in M(q,s)$ full rank matrices with $r\le p$ and $q\ge s$,
$$AXB\sim N_{r,s}(AMB,A\Sigma A^{\top},B^{\top}\mathrm{T}B).
\vspace{-0.1cm}$$
\item[(ii)] If $X_i\sim N_{p,q}(M_i,\Sigma_i,\mathrm{T})$, $i=1,2$, are independent then, 
\vspace{-0.2cm}
$$
X_1+X_2\sim  N_{p,q}(M_1+M_2,\Sigma_1+\Sigma_2,\mathrm{T}).
$$\vspace{-1.2cm}
\item[(iii)] If $X\sim N_{p,q}(0,I_p, T)$ then  $S=X^{\top}X\sim W_q(p,T)$.
    \item[(iv)] If $S\sim W_q(r,T)$ then $S^{-1}\sim W_q^{-1}(r,T^{-1})$. 
\end{itemize}
\end{rem}


\subsection{Upcasting $P^+(q)$ Onto $M(p,q)$}
\label{subsec:upcast}
Let $p\ge q$ and $\U \in P^+(p)$. 
We show that a distribution on 
$P^+(q)$ can be expressed as transform  
of one on the larger space $M(p,q)$ via the surjective mapping $x\mapsto x^{\top} \U ^{-1}x$.

\begin{theo}
\label{th:upcast}
Let $\tilde{\Pi}(\dif S)=\tilde{\pi}(S)\mu(\dif S)$ be a probability measure on $P^{+}(q)$.
Consider the distribution $\Pi(\dif x)$ on $M(p,q)$, $p\ge q$, 
defined as,
\begin{align}
\Pi(\dif x):=\tilde{\pi}(x^{\top}\U ^{-1}x)\,\nu_{\,\U }(\dif x).
\label{eq:change_of_variable}
\end{align}
%
If $X\sim \Pi(\dif x)$ then $S=X^{\top}\U ^{-1}X\sim \tilde{\Pi}(\dif S)$.
\end{theo}
\begin{proof}
%
%
%
\revision{By Equation (\ref{eq:farrell}),}  if $f(x)=g(x^{\top}\U ^{-1}x)$,
for $g:P^+(q)\rightarrow\mathbb{R}$, we have
%
\begin{align*}
\int_{M(p,q)} g(x^{\top}\U ^{-1}x)\nu_{\,\U }(\dif x) 
&=\int_{M(p,q)} g(x^{\top}x)\nu(\dif x)\\
&= 	
\int_{P^+(q)}g(S)\mu(\dif S).
\end{align*}
We have obtained a change of variables formula for map $x\mapsto x^{\top}\U ^{-1}x=S$. 
 %
%
%
Replacing $g(\cdot)$ with 
$(g\cdot \tilde{\pi})(\cdot)$ gives, 
$$
\int_{P^+(q)} g(S)\tilde{\Pi}(\dif S)=\int_{M(p,q)} g(x^{\top}\U ^{-1}x)\Pi(\dif x). 
$$
This completes the proof.
\end{proof}

We will consider MCMC methods on $M(p,q)$. In light of Theorem~\ref{th:upcast}, such algorithms
are directly relevant for distributions $\tilde{\Pi}(\dif S)=\tilde{\pi}(S)\mu(\dif S)$ on $P^+(q)$ since we can execute the MCMC algorithm on $M(p,q)$, with target law $\Pi(\dif x)= 
\tilde{\pi}(x^{\top}\U ^{-1}x)\nu_{\,\U }(\dif x)$, and apply the transform 
$x\mapsto x^{\top}\U ^{-1}x=S$ on the collected $x$-samples.

\section{Matrix-Valued MCMC Methods}
\label{sec:MCMC}

We provide some MCMC methods on $M(p,q)$, $p\ge q$. 
As noted in Section \ref{sec:intro}, we focus on MCMC algorithms involving blind proposals, i.e.\@ containing no information about the target.
 Let $\Pi(\dif x)$ be a target distribution on $M(p,q)$, 
 and -- in a Metropolis setting -- $Q(x,\dif y)$ a proposal  Markov kernel.
 Assume that $\xi$ is a $\sigma$-finite measure such that $\Pi$ is absolutely continuous with respect to  $\xi$, and $Q(x,\dif y)$ is  $\xi$-reversible. 
 %
 %
Then, the triplet $(\Pi,Q, \xi)$ give rise to a Metropolis--Hastings kernel,
so that, for Borel sets $A\subseteq M(p,q)$, \vspace{0.2cm}
\begin{equation*}
P(x,A)=\int_A Q(x,\dif y)\alpha(x,y)+R(x)\cdot\delta_x(A),\vspace{0.2cm}
\end{equation*}
where we have defined $R(x)=1-\int Q(x,\dif y)\alpha(x,y)$, for 
acceptance probability function that has the simple form, 
\begin{equation}
\label{eq:acce}
\alpha(x,y)=\min\{1,\pi(y)/\pi(x)\}, \quad \pi(x):= (\dif\Pi/\dif\xi)(x). 
\end{equation}
%
All Metropolis--Hastings kernels -- on $M(p,q)$ -- in this paper 
will correspond to instances of such triplets $(\Pi,Q, \xi)$. 
A similar notation is adopted for measures restricted
  on $P^{+}(q)$.
%

\subsection{RWM and pCN}
The proposal kernel $Q(x,\dif y)$ of the RWM algorithm on $M(p,q)$ 
is defined by the update
\begin{align}
\label{eq:RWM}
\mathrm{RWM:}\quad  y =x+w, 
\end{align}
with $w\sim N_{p,q}(0, U, V)$,
$U\in P^+(p)$, 
$V\in P^+(q)$. The proposal kernel $Q(x,\,\cdot\,)=N_{p,q}(x, U, V)$ is reversible with respect to the Lebesgue measure on $M(p,q)$. Thus, following the notation we established above, we now have the triplet $(\Pi, Q, \mathrm{Leb}$), and 
the acceptance probability $\alpha(x,y)$ is as in (\ref{eq:acce}), 
%
%
with $\pi(x)$ the density of $\Pi$ with respect to $\mathrm{Leb}$. 
Let $\rho\in[0,1)$.
The pCN method on $M(p,q)$ is determined via the proposal,
\begin{align}
\label{eq:pCN}
\mathrm{pCN:}\quad y =\rho^{1/2}~x+(1-\rho)^{1/2}~w,
\end{align}
with $w\sim N_{p,q}(0, U, V)$.  
The proposal kernel $Q(x,\cdot)=N_{p,q}(\rho^{1/2}x, U, (1-\rho)V)$ is reversible with respect to $N_{p,q}(0, U, V)$. 
We have the triplet $(\Pi, Q, N_{p,q}(0,U,V))$, so the acceptance probability is as in (\ref{eq:acce}), 
where $\pi(x)$ is density of $\Pi$ with respect to $N_{p,q}(0, U,V)$. 
The pCN algorithm was introduced in Section 4.2 of \cite{MR1723510}. As shown there, 
and in several more recent works \citep[see e.g.][]{MR2444507,MR3135540} pCN can be very effective in high dimensions in the context of Gaussian priors
and not highly informative observations. 
The algorithms are well-defined even on infinite-dimensional Hilbert spaces; this has sparked the use of pCN in the area of Bayesian Inverse Problems, see e.g.~the overview in
\cite{stua:10}. 

\subsection{MpCN on $M(p,q)$}

\subsubsection{Derivation via Bayesian Paradigm}
We return briefly at the scenario of Section \ref{subsec:upcast}, i.e.~adopt the viewpoint that 
the original target distribution is given as $\tilde{\Pi}(\dif S)=\tilde{\pi}(S)\mu(\dif S)$ on  
$P^{+}(q)$, and one aims to generate $X\sim \Pi(\dif x)=
\tilde{\pi}(x^{\top}\U ^{-1}x)\nu_{\,\U }(\dif x)$ on $M(p,q)$, $p\ge q$, and return
$S=X^{\top}\U ^{-1}X$. In such an upcasted setting, one can provide a motivation for the derivation of MpCN.
Let $x_{i\cdot}$ denote the $i$th row of $x\in M(p,q)$. Since the density of 
$\Pi(\dif x)$ with respect to $\operatorname{Leb}$ writes as, 
$$f\big(z_{1\cdot}^{\top}\,z_{1\cdot} + \cdots + z_{p\cdot}^{\top} z_{p\cdot}\big); \quad z:= \U ^{-1/2}x,$$
%
where $f(S) = 
\tilde{\pi}(S)\det (S)^{-p/2}$, $S\in P^+(q)$,
it is clear that, under $X\sim \Pi(\dif x)$, 
$$Z_{1\cdot}=_d Z_{2\cdot} =_d \cdots =_d Z_{p\cdot}, $$ 
so all rows of $Z:=\U ^{-1/2}X$ have the same marginal law.
Looking back at the pCN proposal (\ref{eq:pCN}), as adjusted under the linear transform $x\mapsto \U ^{-1/2}x=:z$ , the above understanding 
can provide guidance for tuning the algorithmic parameter $V\in P^{+}(q)$ --  corresponding to the variance of each row-vector of the noise $\U ^{-1/2}w$ -- given information about the current position $z=\U ^{-1/2}x$ of the MCMC chain. In a `classical' approach, 
a likelihood-based choice would simply be the sample variance over the rows, 
$\widehat{V}:=\sum_{i=1}^{p}z_{i\cdot}^{\top}z_{i\cdot}/p$.
A Bayesian approach seems preferable as it will ultimately provide an algorithm with a heavier-tailed proposal. A choice of  \emph{Jeffrey's prior} $\mu(\dif V)$ \citep[see e.g.][]{geis:63}, with $\mu$ as given in (\ref{eq:mmu}), combined with a likelihood 
$\phi_{p,q}(z;0,\U ^{-1/2}U\U ^{-1/2},V)$, $U\in P^{+}(p)$, is easily shown to provide the   posterior  $W_q^{-1}(p,x^{\top}U^{-1}x)$ for $V$.
%
%
%
%
%
%

\subsubsection{MpCN Proposal and Acceptance Probability}
The above thinking gives rise to the following proposal,
applicable for a general law $\Pi(\dif x)$ on $M(p,q)$ -- not only in the upcasted scenario adapted above for purposes of providing some rationale 
under a particular viewpoint --, 
\begin{gather}
\label{eq:mpcn_proposal}
\mathrm{MpCN:}\quad y = \rho^{1/2}x +(1-\rho)^{1/2}w; \\ 
w\sim N_{p,q}(0, U, V);\quad  V\sim W_q^{-1}(p,x^{\top}U^{-1}x). \nonumber
\end{gather}

\begin{rem}
\begin{itemize}
\item[i.] \revision{In the scalar case, it is rather common to choose an Inverse-Gamma distribution for a variance parameter $V\in \mathbb{R}_{+}$, 
thus  an Inverse-Wishart distribution in the matrix case, $V\in P^{+}(q)$, appears as a fairly natural choice.}
\item[ii.] The law of 
$[\,w\,|\,x\,]$ is that of a Matrix-Student-t distribution with Lebesgue density proportional 
to $z\mapsto  \det(x^{\top}U^{-1}x+z^{\top}U^{-1}z)^{-p}$; 
see e.g.~\cite{dick:67}.
\end{itemize}
\end{rem}
\noindent MpCN is a Metropolis method with target  $\Pi(\dif x)$ on $M(p,q)$, and proposal as determined in (\ref{eq:mpcn_proposal}). Let $Q$ denote the related proposal transition kernel.
The lemma that follows shows that $Q$ is $\nu_{\,U}$-reversible, 
so we have the triplet $(\Pi, Q, \nu_{\,U})$
and the acceptance probability is as in (\ref{eq:acce}), 
%
%
%
where $\pi(x)$ is the density of $\Pi$ with respect to $\nu_{\,U}$. 


\begin{lem}\label{lem:kernel_mpcn}
The MpCN proposal kernel $Q$ on $M(p,q)$ has density
$\mathsf{q}(x, y)$ with respect to $\nu_{\,U}(\dif y)$, that writes as, 
\begin{align}
\mathsf{q}(x, y) = (1-\rho)^{pq/2}c_{p,q}
~\frac{\det(x^{\top}U^{-1}x)^{p/2}\times\det(y^{\top}U^{-1}y)^{p/2}}{
	\det(R(x,y))^p},
\label{eq:mpcn_density}
\end{align}
where, 
    $c_{p,q}=\frac{\Gamma_q(p)}{\Gamma_q(p/2)^2}$, 
and, 
\begin{align*}
R(x,y)= &\,x^{\top}U^{-1}x + y^{\top}U^{-1}y-\rho^{1/2}x^{\top}U^{-1}y-\rho^{1/2}y^{\top}U^{-1}x.
\end{align*}	
Therefore, $Q$ is $\nu_{\,U}$-reversible. 
\end{lem}
\begin{proof}
Following the definition of MpCN, the joint distribution of
$[\,V,y\,|\,x\,]$ writes as,
\begin{align}
\label{eq:joint}
W_q^{-1}(\dif V; p,x^{\top}U^{-1}x)~ N_{p,q}(\dif y;\rho^{1/2}x,U,(1-\rho)V).
\end{align}
%
%
It remains to integrate out $V$ in (\ref{eq:joint}), and take the density of
the resulted distribution of $[\,y\,|\,x\,]$ with respect to $\nu_{\,U}(\dif y)$. 
All such calculations can be carried out analytically due to $\mu(\dif V)$ 
being conjugate with respect to Matrix-Normal distribution with
right-covariance matrix $V$. Thus, tedious but otherwise straightforward calculations give the density $\mathsf{q}(x,y)$ in expression (\ref{eq:mpcn_density}).
Since $\mathsf{q}(x,y)=\mathsf{q}(y,x)$, we have $\nu_{\,U}(\dif x)Q(x,\dif y)=\nu_{\,U}(\dif y)Q(y,\dif x)$.
\end{proof}



\subsection{Random-Walk Property of MpCN on $P^{+}(q)$ 
}
\label{sec:mark}

Returning to the context of an initial target $\tilde{\Pi}(\dif S)$ on $P^{+}(q)$, we show here that, when the operators $\U $ and $U$ -- used, respectively, when upcasting the target $\tilde{\Pi}$ onto $M(p,q)$ and as a parameter in the MpCN proposal -- coincide (i.e., $\U=U$), then  
MpCN on $M(p,q)$ induces a Markovian kernel on $P^{+}(q)$ under the transform $x\mapsto x^{\top}U^{-1}x$.
Such kernel 
is easier to analyse as it exhibits random-walk-type behavior, as shown in the proposition and theorem that follow. 
%
%
We define the operation 
$
A\circ B=B^{1/2}AB^{1/2} 
$, $A,B\in P^{+}(q)$;
it follows that $(A\circ B)^{-1}=A^{-1}\circ B^{-1}$ and $\tr(A\circ B)=\tr(AB)$. Recall also the definition of the uniform measure $\dif \mathsf{U}$ in Section \ref{sec:reference}.
%
%

\begin{prop}\label{prop:random_walk}
Let $\tilde{\Pi}(\dif S)=\tilde{\pi}(S)\mu(\dif S)$
and $\Pi(\dif x) =
\tilde{\pi}(x^{\top} U^{-1} x)\nu_{\,U}(\dif x)$, be probability measures on $P^{+}(q)$ and $M(p,q)$ respectively, with $U\in P^+(q)$, $p\ge q$.
Given $x\in M(p,q)$, set $y\sim Q(x,\dif y)$, where $Q=Q(\,\cdot\,;\rho, U)$ is the proposal kernel of MpCN, with parameters $\rho\in[0,1)$ and $U$.
\begin{itemize} 
\item[(i)] We have the representation, for $\det(x^{\top}U^{-1}x)\neq 0$,
$$y^{\top}U^{-1}y=\epsilon\circ(x^{\top}U^{-1}x),$$ 
for some $P^+(q)$-valued random
variable $\epsilon$ with law 
 that does not depend on $x$, $U$. 
\item[(ii)] The law of $\epsilon$, $\mathcal{L}(\dif \epsilon)$, writes as, 
\begin{align}
\label{eq:gamma}
\mathcal{L}(\dif \epsilon)=(1-\rho)^{pq/2}c_{p,q}
~\mu(\dif \epsilon)~\int_{\mathcal{O}(p,q)}\frac{\det(\epsilon)^{p/2}}{
	\det(I_q+\epsilon-\rho^{1/2}(\mathsf{U}_1\epsilon^{1/2}+\epsilon^{1/2}\mathsf{U}_1^{\top}))^p}
	\dif \mathsf{U},
\end{align}
with $c_{p,q}$ as given in Lemma \ref{lem:kernel_mpcn},
where $\mathsf{U}_1$ denotes the first $q$ rows of $\mathsf{U}$.
In particular, $\mathcal{L}(\dif \epsilon)$ depends on $\rho$, but not on $U$. 

\item[(iii)] 
We have $\epsilon =_d 
\epsilon^{-1}$.
\end{itemize}
\end{prop}

\begin{proof}
	Let
	$u_1(x)=U^{-1/2}x(x^{\top}U^{-1}x)^{-1/2}\in\mathcal{O}(p,q)$.
	There exists an orthogonal complement, $u_2(x)\in M(p, p-q)$ as an analytic function,  such that $u(x)=(u_1(x), u_2(x))\in \mathcal{O}(p)$. 
	We consider linear transformations of $y$,  
	\begin{equation}\nonumber
	z_i=u_i(x)^{\top}U^{-1/2}y(x^{\top}U^{-1}x)^{-1/2},\,\,\,\,\, i=1,2,\\
	\end{equation}
	and set $z=(z_1, z_2)\in M(p,q)$. 
	Then $\epsilon=z^{\top}z$. 
	By Lemma \ref{lem:kernel_mpcn}, and the change of variables formula \citep[Lemma 1.5.1 of][]{MR1960435}, 
%
	the law of $z$ is,
\begin{align*}
(1-\rho)^{pq/2}c_{p,q}
~
	\frac{\det(z^{\top}z)^{\revision{p}/2}}{
	\det(I_q+z^{\top}z-\rho^{1/2}(z_1+z_1^{\top}))^p}~\nu(\dif z),
\end{align*}
which does not involve $x$ or $U$. Also, by the change of variables formula (\ref{eq:farrell}), the law of $\epsilon=z^{\top}z$ is as in (\ref{eq:gamma}). 
Finally, the law, $\mathcal{L}_{-1}(\dif\epsilon)$, of the inverse $\epsilon^{-1}$ writes as (upon recalling that $\mu_{-1}\equiv \mu$), 
\begin{align*}
	\mathcal{L}_{-1}(\dif \epsilon) &= (1-\rho)^{pq/2}c_{p,q}
	~\mu(\dif \epsilon)\times
	\int_{\mathcal{O}(p,q)}\frac{\det(\epsilon^{-1})^{p/2}}{
	\det(I_q+\epsilon^{-1}-\rho^{1/2}(\mathsf{U}_{1}\epsilon^{-1/2}+\epsilon^{-1/2}
	\mathsf{U}_{1}^{\top}))^p}\dif \mathsf{U}\\
	&=(1-\rho)^{pq/2}c_{p,q}
	~
	\mu(\dif \epsilon)\times
	\int_{\mathcal{O}(p,q)}\frac{\det(\epsilon)^{p/2}}{
	\det(I_q+\epsilon-\rho^{1/2}(\mathsf{U}_{1}^{\top}\epsilon^{1/2}+\epsilon^{1/2}
	\mathsf{U}_{1}))^p}\dif \mathsf{U}.
\end{align*}
%
From Theorem 3.3.1 of \cite{MR1960435}, the uniform distribution on $\mathcal{O}(p,q)$ is the marginal (on $\mathcal{O}(p,q)$) of the uniform distribution on $\mathcal{O}(p)$. Also, by Section 1.4.1 of \cite{MR1960435}, the uniform distribution on $\mathcal{O}(p)$ is invariant under the matrix transpose operation.
Thus, the distribution of $\mathsf{U}_1$ is also invariant under the matrix transpose operation. 
The proof is now complete. 
\end{proof}
\noindent Proposition \ref{prop:random_walk} leads to the theorem below.
\begin{theo}
\label{th:Markov}
Let 
$\tilde{\Pi}(\dif S)=\tilde{\pi}(S)\mu(\dif S)$ 
be a target on 
$P^{+}(q)$.
Define the corresponding upcasted law on $M(p,q)$, $\Pi(\dif x) = 
\tilde{\pi}(x^{\top}U^{-1}x)\nu_{\,U}(\dif x)$, $U\in P^{+}(q)$, and let $Q = Q(\,\cdot\,;\rho,U)$ be the MpCN proposal kernel. Consider also the kernel 
$\tilde{Q}(S,\,\cdot\,):=_{d}\epsilon\,\circ\, S$, $S\in P^+(q)$, with the law $\mathcal{L}(\dif \epsilon)$ of $\epsilon\in P^{+}(q)$ as  determined in Proposition \ref{prop:random_walk}(ii), for parameters $\rho\in[0,1)$ and $p$, $p\ge q$. \\
If $\{X_n\}_{n\ge 0}$ is the MpCN Markov chain with target $\Pi$ and proposal $Q$,   
 then the process $\{S_{n}\}_{n\ge 0}$ with,
\begin{align*}
S_n:= X_n^{\top}U^{-1}X_n,
\end{align*}
is a Metropolis--Hastings Markov chain, with respect to its own filtration, with target 
$\tilde{\Pi}$ and proposal $\tilde{Q}$ (and initial position $S_0 =_d X_0^{\top}U^{-1}X_0$).
Moreover, it is a \emph{Random-Walk} Markov chain, in the sense 
that $\tilde{Q}$ is the transition kernel of a random walk 
$T_n= \epsilon_n \circ T_{n-1}$, $n\ge 1$, 
where $\epsilon_n$ follows a probability distribution $\mathcal{L}$ satisfying $\mathcal{L}=\mathcal{L}_{-1}$.
\end{theo}
%
%
%
%

Given Theorem \ref{th:Markov}, we will refer -- without confusion -- 
to the Markov kernel, denoted $\tilde{P}$,  with target $\tilde{\Pi}$ and proposal $\tilde{Q} = \tilde{Q}(\,\cdot\,;\rho,p)$ as ``an MpCN kernel on $P^{+}(q)$". Notice that both $\tilde{\Pi}$ and $\tilde{Q}$ do not depend on the choice of parameter $U$. The MpCN kernel on $P^+(q)$ is a Random-Walk Metropolis kernel on this space.  





\begin{proof}[Proof of Theorem \ref{th:Markov}]
We show that $\{S_n\}_{n\ge 0}$ is in the class of Metropolis--Hastings Markov chains introduced in the beginning of Section \ref{sec:MCMC}, with triplet $(\tilde{\Pi},\tilde{Q},\mu)$.  
First, we prove $\mu$-reversibility of $\tilde{Q}$, by making use of the $\nu_U$-reversibility of $Q$ itself. Indeed, for any given Borel sets $A, B\subseteq P^+(q)$, upon defining the sets 
$A^*=\{x\in M(p,q):x^{\top}U^{-1}x\in A\}$ and 
$B^*=\{x\in M(p,q):x^{\top}U^{-1}x\in B\}$, we have,
\begin{align*}
\int_A\tilde{Q}(S,B)\mu(\dif S)&=
\int_{A^*}Q(x,B^*)\nu_U(\dif x)\\
&=
\int_{B^*}Q(x,A^*)\nu_U(\dif x)= \int_B\tilde{Q}(S,A)\mu(\dif S). 
\end{align*}
The acceptance probability of $\{X_n\}_{n\ge 0}$ coincides with that of the triplet $(\tilde{\Pi},\tilde{Q},\mu)$ via the surjective mapping $x\mapsto x^{\top}U^{-1}x$. 
The proof is complete. 
\end{proof}


\begin{prop}\label{prop:inv}
%
Let $\{S_n\}_{n\ge 0}$ be a Markov process on $P^{+}(q)$ with transition kernel corresponding to that of an MpCN chain with 
target $\tilde{\Pi}$ and proposal $\tilde{Q}(\,\cdot\,;\rho,p)$.
Then, $\{S_n^{-1}\}_{n\ge 0}$ is a Markov process on $P^{+}(q)$ with transition kernel corresponding to that of an MpCN transition kernel with 
target $\tilde{\Pi}_{-1}$ and identical proposal $\tilde{Q}(\,\cdot\,;\rho,p)$.
\end{prop}


\begin{proof}
Recall that $\{S_n\}_{n\ge 0}$ corresponds to the triplet $(\tilde{\Pi},\tilde{Q},\mu)$. 
We show that $\{S_n^{-1}\}_{n\ge 0}$ corresponds to triplet $(\tilde{\Pi}_{-1}, \tilde{Q}, 
\mu)$. 
Since $S\mapsto S^{-1}$ is a bijection in $P^+(q)$, $\{S_n^{-1}\}_{n\ge 0}$ is also a Markov chain. 
The third component of the triplet is invariant under this transform, since  $\mu_{-1}=\mu$. The second component is also invariant since $\tilde{Q}(S, A)=\tilde{Q}(S^{-1}, A^{-1})$, for $A\subseteq P^+(q)$, because $\epsilon=_d\epsilon^{-1}$ by Proposition \ref{prop:random_walk}(iii).  
The first component of the triplet $\tilde{\Pi}$
becomes $\tilde{\Pi}_{-1}$. 
Finally, the acceptance probability $\alpha(S, S^*)$ of the triplet $(\tilde{\Pi},\tilde{Q},\mu)$ is the same as that of $(\tilde{\Pi}_{-1}, \tilde{Q}, 
\mu)$ via the projection $S\mapsto S^{-1}$. 
Therefore  $\{S_n^{-1}\}_{n\ge 0}$ is the Metropolis chain determined by $(\tilde{\Pi}_{-1}, \tilde{Q}, 
\mu)$. 
The proof is complete.
\end{proof}

Notice that if a Markov chain $\{S_n\}_{n\ge 0}$ is geometrically ergodic, and $T_n=f(S_n)$ forms a Markov chain -- for a map $f$, on appropriate domains -- then $\{T_n\}_{n\ge 0}$ is also geometrically ergodic, since the $\sigma$-algebra generated by $\{S_n\}_{n\ge 0}$ contains that of $\{T_n\}_{n\ge 0}$. 
This observation is used in the following two statements.
\begin{itemize}
\item[(i)]
If the upcasted MpCN kernel on $M(p,q)$ is geometrically ergodic, then the deduced MpCN kernel on $P^{+}(q)$ is also geometrically ergodic. 
\item[(ii)] Proposition \ref{prop:inv} implies that, if the MpCN kernel 
on $P^{+}(q)$ targeting $\tilde{\Pi}$ is geometrically ergodic,  then the MpCN kernel targeting  $\tilde{\Pi}_{-1}$ (for the same $\rho$, $p$)  is also geometrically ergodic. 
Indicatively, geometric ergodicity of MpCN with Wishart target $\tilde{\Pi}\equiv W_q(r,T)$ is equivalent to geometric ergodicity of MpCN with Inverse-Wishart target $\tilde{\Pi}_{-1}\equiv W_q^{-1}(r,T^{-1})$, for $r\in\RR$, $r>q-1$, and $T\in P^{+}(q)$. 
\end{itemize}




\section{Ergodicity Results for RWM and pCN}
\label{sec:ergodicity}

The ergodic properties of the RWM and pCN kernels on a vector space have been studied in \cite{MT2, RT, MR1731030, Rudolf_2016, kama:17}. In this section, we investigate ergodicity on a matrix space. First, we note that RWM and pCN kernels are ergodic under fairly general assumptions \citep[e.g.,][]{MT1994, Kulik_2015}. Therefore, in this section we concentrate on geometric ergodicity. 

Throughout this section, and unless specified otherwise, $U\in P^+(p)$, $V\in P^+(q)$, $\rho\in[0,1)$ are 
the parameters appearing in the RWM and pCN proposal kernels on $M(p,q)$, $p\ge q$.
\revision{In this paper, a Markov kernel $P$ on $(E,\mathcal{E})$ is said to be geometrically ergodic if there is a probability measure $\Pi$ and $r\in (0,1)$ such that 
$$
V(x)=r^{-n}\|P^n(x,\cdot)-\Pi\|_{\mathrm{TV}}
$$
is $\Pi$-integrable, where for general measures $\mu$, $\nu$, 
$\|\mu-\nu\|_{\mathrm{TV}}=\sup_{A\in\mathcal{E}}|\mu(A)-\nu(A)|.$ 
}

\subsection{Ergodicity for RWM}

\noindent 
We provide a sufficient and a necessary condition for the geometric ergodicity of the RWM kernel.
Many MCMC methods do not work well for target distributions with contour manifolds that degenerate -- in a proper sense -- in the tails; 
see, e.g., Section 5 of \cite{MR1731030}.   
To exclude such cases, we consider the 
following class of functions.

\begin{defi}
A continuously differentiable function $h:M(p,q)\mapsto \mathbb{R}_+$ satisfies the contour condition if,
\begin{equation}\label{eq:contour_condition}
	\lim_{\|x\|_F\rightarrow+\infty}\big\langle\tfrac{D h(x)}{\|D h(x)\|_F},~\tfrac{x~}{\|x\|_F}\big\rangle_F<0. 
\end{equation}
\end{defi}

\noindent 




A sufficient condition for geometric ergodicity is formulated by using the above contour condition together with the following  exponentially light tail condition (\ref{eq:right-tail}). \revision{Though the theory in \cite{MR1731030} concerns only vector-valued processes, we can apply their result if we consider the $M(p,q)$ space as a vector space of length $pq$.  }

\begin{prop}[\citealt{MR1731030}]
\label{prop:rwm_GE}
Consider the law $\Pi(\dif x) = \pi(x) \operatorname{Leb}(\dif x)$ on $M(p,q)$. Assume that $\revision{\log}\,\pi$ is continuously differentiable, satisfies the contour condition, and, 
	\begin{equation}\label{eq:right-tail}
		\lim_{\|x\|_F\rightarrow+\infty}\big\langle D\log\pi(x), \tfrac{x~}{\|x\|_F}\big\rangle =-\infty. 
	\end{equation}
	Then, RWM with target $\Pi$ 
is geometrically ergodic. 
\end{prop}
\noindent 
A necessary condition can be formulated via a moment requirement, \revision{that first appeared in \cite{MR1996270} for Euclidean spaces. In \cite{kama:17} such results were generalized to  metric spaces including the matrix space $M(p,q)$ with Frobenius norm $\|\cdot\|_F$. The result can also be applied to probability measures on $P^+(q)$, upcasted onto $M(p,q)$. }

\begin{prop}[\citealt{MR1996270,kama:17}]\label{prop:necessary-condition}\ \\ \vspace{-0.7cm}
\begin{itemize}
\item[(i)]
	If the RWM kernel with target $\Pi(\dif x)$ on $M(p,q)$ 
is geometrically ergodic, then, for some $s>0$,
	$$
	\int_{M(p,q)} \exp(s\|x\|_F)\Pi(\dif x)<\infty.
	$$
\item[(ii)]	
Let $\tilde{\Pi}(\dif S)=\tilde{\pi}(S)\mu(\dif S)$ be a probability distribution on $P^+(q)$, with corresponding upcasted law on  $M(p,q)$, $\Pi(\dif x) =
\tilde{\pi}(x^{\top}\U ^{-1}x)\nu_{\,\U }(\dif x)$. 
	If the RWM chain with target $\Pi$ 
is geometrically ergodic, then, for some $s>0$,
	$$
	\int_{P^+(q)} \exp\left(s~\sqrt{\tr(S)}\right)\tilde{\Pi}(\dif S)<\infty.
	$$
\end{itemize} 
\end{prop}

\noindent

For example, $\alpha\ge 1$ is necessary and $\alpha>1$ sufficient for the Lebesgue density $\pi(x)\propto \exp(-\beta\|x\|_F^\alpha)$, $\beta>0$.  In contrast,  RWM is expected  not to be geometrically ergodic when $\pi(x)$ has heavy tails. 

Regular variation  on $\mathbb{R}$ is a concept met in 
a large literature -- see, e.g., \cite{opac-b1134644, MR1015093} -- and can be used to characterise popular classes of heavy-tailed functions. Regular variation on $P^+(q)$ is a natural extension from $\mathbb{R}$.
We first define an appropriate metric.
%
%
A matrix $A\in P^+(q)$ writes as
$A= \mathsf{U}\Lambda \mathsf{U}^{\top}$, for $\mathsf{U}\in \mathcal{O}(q)$ and 
positive diagonal matrix $\Lambda$. 
The logarithmic map, $\log:P^+(q)\rightarrow \mathrm{Sym}(q)$, is given by 
$
\log A= \mathsf{U}(\log \Lambda) \mathsf{U}^{\top}
$
where $\log \Lambda:=\mathrm{diag}\{\log\lambda_1,\ldots,\log\lambda_q\}$. 
We consider a metric $\mathsf{d}(\cdot,\cdot)$ on $P^+(q)$, such that,   
\begin{equation}
\mathsf{d}(A,B):=\|\log(A\circ B^{-1})\|_F=\big\{\sum_{i=1}^q \log^2\lambda_i\big\}^{1/2},
\label{eq:metric}
\end{equation}
where $\lambda_1,\ldots,\lambda_q$ are eigenvalues of $A\circ B^{-1}$, $A, B\in P^{+}(q)$. 
This is a natural distance induced by the logarithmic map; see Theorem XII.1.3 of \cite{MR1666820}. 
Note that a matrix can have several square roots, in general, but the value of $\mathsf{d}(A,B)$ is unique. 
The topology induced by the metric $\mathsf{d}$ and that by the Frobenius norm are different; the former fits naturally to $P^+(q)$. $P^+(q)$ is a complete metric space under $\mathsf{d}$, but not under the Frobenius norm: e.g., observe that $A_n=\mathrm{diag}\{n^{-1},1,\ldots, 1\}$ forms a Cauchy sequence under the Frobenius norm, with a limit $A=\mathrm{diag}\{0,1,\ldots,1\}\notin P^+(q)$. 

\begin{defi}
\label{def:RV}
 A  function $h:P^+(q)\rightarrow\mathbb{R}_+$ is regularly varying if 
 there exists $r\in\RR$, with $r>q-1$, such that,
	$$
	\frac{h(z x)}{h(z I_q)}\longrightarrow_{z\rightarrow +\infty}\ \det(x)^{\revision{-}r/2},
	$$
	locally uniformly in $x\in P^+(q)$ under the topology induced by the metric $\mathsf{d}(\cdot,\cdot)$ in (\ref{eq:metric}). 
\end{defi}

\begin{rem} 
\label{rem:RV}
The probability density function of the Inverse-Wishart law $W_q^{-1}(r,T)$ is regularly varying since,
\begin{align*}
    \frac{\pi(zx)}{\pi(zI_q)}=\det(x)^{-r/2}\exp\left(-z^{-1}\tr\left(T(x^{-1}-I_q)\right)\right)\longrightarrow_{z\rightarrow +\infty}\ \det(x)^{-r/2}. 
\end{align*}
\end{rem}

\begin{cor}\label{cor:kama17}
Let $\tilde{\Pi}(\dif S) = \tilde{\pi}(S)\mu(\dif S)$ be a probability law on $P^+(q)$. If $\tilde{\pi}$ is continuous, regularly varying, then the RWM chain with target $\Pi(\dif x) = \tilde{\pi}(x^{\top}\U ^{-1}x)\nu_{\,\U }(\dif x)$ on $M(p,q)$ 
is not geometrically ergodic. 
\end{cor}
\begin{proof}
%
Any $S\in P^+(q)$ writes as 
$S= \mathsf{U}^{\top}\Lambda \mathsf{U}= (\sum_{i}{\lambda_i})~\mathsf{U}^{\top}\Lambda^{\ast}\mathsf{U}= z\mathsf{U}^{\top}\Lambda^{\ast}\mathsf{U}$, with
$z=\sum_{i}{\lambda_i}$,
$\mathsf{U}\in\mathcal{O}(q)$, $\Lambda=\mathrm{diag}\{\lambda_1,\cdots,\lambda_q\}$, $\Lambda^{\ast}=\mathrm{diag}\{\lambda_1^{\ast},\cdots,\lambda_q^{\ast}\}$,  $0<\lambda_1<\cdots<\lambda_q$, where we have set $\lambda_i^* = \lambda_i/\sum_j \lambda_j$.
%
We define a bounded set, for $\varepsilon \ge 0$,
$$
\Gamma_\varepsilon:=\Big\{\lambda^*=(\lambda_1^*,\ldots, \lambda_{q-1}^*)\in\mathbb{R}_+^{q-1}: \varepsilon<\lambda_1^*<\cdots<\lambda_q^*,\, \lambda_q^* = 1-\sum_{i=1}^{q-1}\lambda_i^*\Big\}. 
$$
We introduced $\varepsilon$ here so that we can use Lebesgue's dominated convergence theorem in the following inequality. 
By Theorem 5.3.1 of \cite{MR770934}, for $f:\mathbb{R}_+\rightarrow\mathbb{R}_+$, we have,
\begin{align*}
\int_{S\in P^+(q)} &f(\tr(S))\tilde{\Pi}(\dif S)=c\int f(\tr\Lambda)\prod_{i<j}(\lambda_j-\lambda_i)\prod_{i=1}^q\lambda_i^{-(q+1)/2}\dif\lambda_i\int_{\mathcal{O}(q)} \tilde{\pi}(\mathsf{U}^{\top}\Lambda\mathsf{U})\dif\mathsf{U}\\
&=c\int_{\mathbb{R}_{+}} f(z)\Big\{\int_{\lambda^*\in\Gamma_0}\prod_{i=1}^{q}\prod_{j=i+1}^{q}(\lambda_j^*-\lambda_i^*)(\lambda_i^*)^{-(q+1)/2}\dif\lambda^{\ast}\int_{\mathcal{O}(q)} \tilde{\pi}(z\mathsf{U}^{\top}\Lambda^{\ast}\mathsf{U})\dif\mathsf{U}\Big\}z^{-1}~\dif z\\
&\ge c\int_{\mathbb{R}_{+}} f(z)\Big\{\int_{\lambda^{\ast}\in\Gamma_\varepsilon}\prod_{i=1}^{q}\prod_{j=i+1}^{q}(\lambda^{\ast}_j-\lambda^{\ast}_i)(\lambda^{\ast}_i)^{-(q+1)/2}\dif\lambda^{\ast}\int_{\mathcal{O}(q)} \tilde{\pi}(z\mathsf{U}^{\top}\Lambda^{\ast}\mathsf{U})\dif\mathsf{U}\Big\}z^{-1}~\dif z\\
&=c\int_{\mathbb{R}^+} f(z) h_\varepsilon(z) z^{-1}\dif z, 
\end{align*}
with $h_\varepsilon(z)$ defined above in an obvious way.
Given that $\tilde{\pi}$ is assumed to be regularly varying, $h_\varepsilon(z)$ can also be shown to be regularly varying, since for any $x>0$ we have,
\begin{align*}
      \frac{\tilde{\pi}(zx\mathsf{U}^{\top}\Lambda^{\ast}\mathsf{U})}{\tilde{\pi}(zI_q)}\longrightarrow_{z\rightarrow\infty}\det(x~\mathsf{U}^{\top}\Lambda^*\mathsf{U}))^{-r/2}=x^{-rq/2}\det(\mathsf{U}^{\top}\Lambda^*\mathsf{U})^{-r/2},
\end{align*}
for some $r\in\RR$, with $r>q-1$;
hence, from Lebesgue's dominated convergence theorem,
\begin{align*}
      \frac{h_\varepsilon(zx)}{h_\varepsilon(z)}=\frac{h_\varepsilon(zx)/\tilde{\pi}(zI_q)}{h_\varepsilon(z)/\tilde{\pi}(zI_q)}\longrightarrow_{z\rightarrow\infty} x^{-rq/2}.
\end{align*}
Thus, for 
$f(x)=\exp(s\sqrt{x})$, $s,x>0$, we have obtained, 
\begin{align*}
\int_{S\in P^+(q)} \exp\big(s\sqrt{\tr(S)}~\big)\tilde{\Pi}(\dif S)
	&\ge c\int_0^\infty \exp\left(s\sqrt{z}~\right)h_\varepsilon(z)~z^{-1}~\dif z. 
\end{align*}
By Theorem 1.5.6(iii) of \cite{MR1015093},  $h_\varepsilon(z)z^{-1}\ge Cz^{-rq/2-1-\delta}$,
for any $\delta>0$, for some $C>0$, and sufficiently large $z>0$. The claim of Corollary \ref{cor:kama17} follows, since,
\[
\lim_{z\rightarrow\infty}\exp(s\sqrt{z}~)h_\varepsilon(z)z^{-1}=+\infty, 
\]
for any $s>0$, that violates the integrability condition in Proposition \ref{prop:necessary-condition}(ii). 
\end{proof}
%
%
\subsubsection{Example Cases}
\label{sec:example-rwm}
Using the above results, we check geometric ergodicity for RWM for the standard probability measures written down in Section \ref{subsec:pm}. 
\begin{enumerate}
    \item 
Let $\pi(x)$ be the density function of $N_{p,q}(M,\Sigma,\mathrm{T})$ under $\mathrm{Leb}$. Then, we have the derivative $D\log\pi(x)=-\Sigma^{-1}(x-M)\mathrm{T}^{-1}$. Thus, as $\|x\|_F\rightarrow \infty$, the inner product term at the contour condition (\ref{eq:contour_condition}) is dominated from above by,
$$
-\inf_{e:\|e\|_F=1}\frac{\left\langle \Sigma^{-1}e\mathrm{T}^{-1},e\right\rangle_{F}}{\|\Sigma^{-1}e \mathrm{T}^{-1}\|_F}
\le 
-\frac{\inf_{e:\|e\|_F=1}\left\langle \Sigma^{-1}e \mathrm{T}^{-1},e\right\rangle_F}{\sup_{e:\|e\|_F=1}\|\Sigma^{-1}e \mathrm{T}^{-1}\|_F}. 
$$
The denominator is bounded from above, and the numerator is bounded away from~$0$ since
$\Sigma$ and $\mathrm{T}$ are positive definite. 
Thus, the density satisfies the contour condition. We can also check that $\pi$ satisfies  (\ref{eq:right-tail}). Therefore, RWM is geometrically ergodic. 
\item Consider the Wishart distribution $W_q(r,T)$, $r\in \RR$, $r>q-1$. 
Since the probability measure is defined on $P^+(q)$, we use the upcasting strategy to apply RWM. The  density function of the upcasted law of $W_q(r,T)$ under $\mathrm{Leb}$ is,
\begin{align*}
    \pi(x)= const.\times
    \det (x^{\top}\U ^{-1}x)^{r/2}\exp\left(-\tr\,[\,T^{-1}(x^{\top}\U ^{-1}x)\,]/2\right)\,\det(x^{\top}\U ^{-1}x)^{-p/2}.
\end{align*}
Thus,
\begin{align*}
    D\log\pi(x)    &=\frac{r-p}{2} D\log(\det(x^{\top}\U ^{-1}x))
    -\U ^{-1}x T^{-1}. 
\end{align*}
By the triangle inequality, inequalities (\ref{eq:contour_condition}) and (\ref{eq:right-tail}) will follow once we show,
$$
\lim_{\|x\|_F\rightarrow\infty}\frac{\|D\log(\det(x^{\top}\U ^{-1}x))\|_F}{\|x\|_F}=0. 
$$
By Equation (15.8.6) of \cite{MR1467237}, 
\begin{align*}
    \frac{\dif}{\dif x_{ij}}\log(\det(x^{\top}\U ^{-1}x))
    &=\tr\Big[\,(x^{\top}\U ^{-1}x)^{-1}\frac{\dif}{\dif x_{ij}}(x^{\top}\U ^{-1}x)\,\Big]\\
    &=2~[~\U ^{-1}x(x^{\top}\U ^{-1}x)^{-1}~]_{ij}. 
\end{align*}
Therefore, 
\begin{align*}
    \|D\log(\det(x^{\top}\U ^{-1}x))\|_F^2
    &=4\sum_{ij}[~\U ^{-1}x(x^{\top}\U ^{-1}x)^{-1}~]_{ij}^2\\
    &\le 4\|x\|_F^{-1}\sup_{e:\|e\|_F=1}\sum_{ij}[~\U ^{-1}e(e^{\top}\U ^{-1}e)^{-1}~]_{ij}^2,
\end{align*}
which converges to $0$ as $\|x\|_F\rightarrow\infty$. 
So inequalities (\ref{eq:contour_condition}) and (\ref{eq:right-tail}) hold. This implies that RWM is geometrically ergodic for the Wishart distribution. 
\item
By Corollary \ref{cor:kama17}, 
since the Inverse-Wishart distribution has regularly varying density function (see Remark \ref{rem:RV}), RWM is not geometrically ergodic for $W^{-1}_q(r,T)$. 
\end{enumerate}
%
%
\subsection{Ergodicity of pCN kernel}

The pCN algorithm is more sensitive as to the choice of target distribution. We establish a necessary condition for geometric ergodicity by using the decay of deviation of $\log\pi$, with $\pi$ the Lebesgue density of the target. By the following result, if $\pi(x)\propto \exp(-\beta\|x\|_F^\alpha)$, $x\in M(p,q)$, then it is  necessary that $\alpha\ge 2$.


\begin{prop}\,\,
\label{prop:strange}
\begin{itemize}
    \item[(i)]
    Consider $\Pi(\dif x)=\pi(x)\,\mathrm{Leb}(\dif x)$.
    Let  $\lambda$ be the largest amongst all eigenvalues of $U$, $V$ (recall these are parameters appearing at the pCN proposal kernel).
    Set,
    \[
    C_{r,\varepsilon}(s):=\sup_{s':\|s-s'\|_F\le \varepsilon}r^{-2}\big|\,\log\pi(rs)-\log\pi(rs')\,\big|.
    \]
    If, for some $\varepsilon>1-\rho^{1/2}$,
    \[
    \inf_{s:\|s\|_F=1}\liminf_{r\rightarrow\infty}C_{r,\varepsilon}(s)<\tfrac{1-\rho}{2}\lambda^{-2},
    \]
    then the pCN kernel with target $\Pi$ is not geometrically ergodic. 
    \item[(ii)]
    Consider $\tilde{\Pi}(\dif S)=\tilde{\pi}(S)\mu(\dif S)$. 
    Let  $\tilde{\lambda}$ be the largest amongst all eigenvalues of $U$ and $V$
    (recall\, $\U $ is the matrix appearing in the definition of the upcasted target $\Pi$). 
    Set,
    \[
    C_{r,\varepsilon}(S):=\sup_{S':\mathsf{d}(S,S')<\varepsilon}r^{-2}
    \big|\,\log\tilde{\pi}(r^2S)-\log\tilde{\pi}(r^2S')\,\big|. 
    \]
    If, for some $\varepsilon>q^{1/2}2\log(1-\rho^{1/2})$,
    \[
    \inf_{S\in P^+(q), \tr(S)=1}\liminf_{r\rightarrow\infty}
    C_{r,\varepsilon}(S)<\tfrac{1-\rho}{2}\tilde{\lambda}^{-2},
    \]
    then the pCN chain with target $\Pi$ is not geometrically ergodic.
\end{itemize}
\end{prop}
See Appendix \ref{app:strange} for the proof. 
We state an immediate consequence of Proposition \ref{prop:strange} for probability measures with regularly varying densities on $P^+(q)$. 

\begin{cor}\label{cor:pcn-regular}
Consider the target $\tilde{\Pi}(\dif S)=\tilde{\pi}(S)\mu(\dif S)$ on $P^+(q)$. 
If $\tilde{\pi}$ is continuous, regularly varying, then the pCN kernel with the corresponding upcasted target $\Pi(\dif x)$ on $M(p,q)$ is not geometrically ergodic.
\end{cor}

\begin{proof}
By the regularly varying property, for some $r\in\RR$, $r>q-1$,
\begin{align*}
\sup_{S':\mathsf{d}(S,S')<\varepsilon}\big|\,\log\tilde{\pi}(rS)-\log\tilde{\pi}(r S')\,\big|
&=
\sup_{S':\mathsf{d}(S,S')<\varepsilon}\Big|\,\log\frac{\tilde{\pi}(rS)/\tilde{\pi}(rI_q)}{\tilde{\pi}(r S')/\tilde{\pi}(rI_q)}\,\Big|\\
&\quad \longrightarrow_{r\rightarrow\infty}
\sup_{S':\mathsf{d}(S,S')<\varepsilon}\Big|\,\log\frac{(\det S)^{-r/2}}{(\det S')^{-r/2}}\,\Big|<\infty. 
\end{align*}
In particular, 
\begin{align*}
\sup_{S':\mathsf{d}(S,S')<\varepsilon}r^{-1}|\,\log\tilde{\pi}(rS)-\log\tilde{\pi}(r S')\,|\longrightarrow_{r\rightarrow\infty} 0. 
\end{align*}
Thus, the pCN kernel is not geometrically ergodic. 
\end{proof}


\subsubsection{Example Cases}
We check geometric ergodicity for pCN for the standard targets in Section \ref{subsec:pm}. We only state negative results as we did not have positive results in this paper. 
\begin{enumerate}
\item Let $\pi(x)$ be the density function of $N_{p,q}(M, \Sigma, \mathrm{T})$. Let \revision{$\lambda$ and} $\lambda_{\star}$ be the \revision{largest and} smallest amongst all eigenvalues of $\Sigma$ and $\mathrm{T}$.  Then,
\begin{align*}
    C_{r,\varepsilon}(s)
    &\rightarrow_{r\rightarrow\infty}\sup_{s':\|s-s'\|_F<\varepsilon}
    \left|\,\tr\left[\,\mathrm{T}^{-1}s^{\top}\Sigma^{-1}s-\mathrm{T}^{-1}s'^{\top}\Sigma^{-1}s'\,\right]/2\,\right|\\
    &=\sup_{s':\|s-s'\|_F<\varepsilon}
    \left|\,\left\langle (s-s')\mathrm{T}^{-1},\Sigma^{-1}(s+s')\right\rangle_F/2\,\right|\\
    &\le \lambda_{\star}^{-2}\sup_{s':\|s-s'\|_F<\varepsilon}\big\{\,\|s-s'\|_F~\|s+s'\|_F/2\,\big\}\\
    &\le \lambda_{\star}^{-2}\varepsilon(\|s\|_F+\varepsilon/2). 
\end{align*}
Thus, if,
$$
\lambda_{\star}^{-2}(1-\rho^{1/2})(1+(1-\rho^{1/2})/2)
< \tfrac{1-\rho}{2}\lambda^{-2}, 
$$
which is simplified to,
$$
\frac{3-\rho^{1/2}}{1+\rho^{1/2}}
< \frac{\lambda_{\star}^2}{\lambda^2},
$$
then pCN is not geometrically ergodic. Therefore,  pCN is not \emph{always} geometrically ergodic for the matrix-normal distribution. 
\item 
Similar calculations yield that, for a Wishart distribution $W_q(r,T)$, we have,
$$
C_{r,\varepsilon}(S)\longrightarrow_{r\rightarrow\infty} \sup_{S':\mathsf{d}(S,S')<\varepsilon}
\left|\,\tr\,[\,T^{-1}(S-S')\,]\,\right|/2. 
$$
Therefore, if $T$ is large enough, then 
the condition in Proposition \ref{prop:strange}(ii) is satisfied, and pCN is not geometrically ergodic. 
    \item pCN is not geometrically ergodic for the upcasted Inverse-Wishart distribution by Corollary \ref{cor:pcn-regular}, since the relevant probability density function is a regularly varying function (see Remark \ref{rem:RV}). 
\end{enumerate}




\section{MpCN Ergodicity Investigation}
\label{sec:mpcn-ergodicity}

Consider a target law $\Pi(\dif x) = \pi(x) \nu_{\U }(\dif x)$.
If $\pi(x)>0$ for any full-rank $x\in M(p,q)$ then MpCN is ergodic. The MpCN algorithm is expected to have good convergence properties even for heavy-tailed target distributions.
We have made a lot of progress towards proving \emph{geometric} ergodicity for MpCN, via use of  Foster--Lyapunov drift criterion \citep{MT} and Dirichlet form. 
The last remaining steps for a fully completed proof remain subject of future research.

\subsection{Our Results}

Our investigation focuses mainly at the setting of Section \ref{sec:mark}, when: the initial target is
$\tilde{\Pi}(\dif S) = \tilde{\pi}(\dif S)\mu(\dif S)$, $S\in P^{+}(q)$;
the upcasted distribution is $\Pi(\dif x ) = \tilde{\pi}(x^{\top}\U ^{-1}x)\nu_{\,\U }(\dif x)$, 
$x\in M(p,q)$; the choice $\U=U$ gives rise to a Markov chain on $P^{+}(q)$ defined via what we have called the MpCN kernel  $\tilde{P}$ on $P^+(q)$, with target $\tilde{\Pi}$ and proposal kernel $\tilde{Q}=\tilde{Q}(\,\cdot\,;\rho,p)$ both of which do not depend on $U$ -- see also the comment after Theorem \ref{th:Markov}. 
%
%
%
%
%
%
%
%
%
To prove geometric ergodicity, it suffices to show, that 
for drift function $\mathsf{V}:P^{+}(q)\rightarrow \mathbb{R}_+$,
\begin{gather}
    \limsup_{\tr(S)\rightarrow\infty}\tfrac{\tilde{P}\mathsf{V}(S)-\mathsf{V}(S)}{\mathsf{V}(S)}<0;
    \label{eq:drift1}
    \\
    \limsup_{\det(S)\rightarrow 0}\tfrac{\tilde{P}\mathsf{V}(S)-\mathsf{V}(S)}{\mathsf{V}(S)}<0.
    \label{eq:drift2}
\end{gather}
\revision{We provide here a brief explanation for the requirement to consider the limits $\tr(S)\rightarrow \infty$, $\det(S)\rightarrow 0$, and not the ones $\|S\|_F\rightarrow \infty$, $\|S\|_F\rightarrow 0$ that might appear as more natural candidates. Since $S\mapsto\det(S)$, $S\mapsto \tr(S)$ are continuous in the metric $\mathsf{d}(A, B)$ in (\ref{eq:metric}),  $\{S\in P^+(q):\epsilon\le \det(S),\ \tr(S)\le \epsilon^{-1}\}$  is a compact set for $\epsilon\in (0,1)$. On the other hand, every compact set is a small set for $\tilde{P}$. Thus, (\ref{eq:drift1}) and (\ref{eq:drift2}) give that $(\tilde{P}\mathsf{V}-\mathsf{V})/\mathsf{V}$ is smaller than a negative constant outside a small set, implying that the geometric drift condition is satisfied. In contrast, $\{S\in P^+(q):\epsilon\le \|S\|_F\le \epsilon^{-1}\}$ is \emph{not} a compact set in the metric $\mathsf{d}(A, B)$. Thus, indeed one needs to work with the limits $\tr(S)\rightarrow \infty$, $\det(S)\rightarrow 0$, and not with $\|S\|_F\rightarrow \infty$, $\|S\|_F\rightarrow 0$. }
We prove (\ref{eq:drift1}), for an appropriate drift function.
%
We start with a definition (for matrices $S, T\in P^+(q)$,
we write $S<T$ if  $T-S\in P^+(q)$).


\begin{defi}
\label{def:RAV}
\begin{itemize}
    \item[(i)] 
	For $T\in P^+(p)$, $h:M(p,q)\rightarrow\mathbb{R}$ is rapidly varying if, for full rank $x, y\in M(p,q)$,  and 
	$z\in P^+(q)$,
	$$
	\frac{h(yz)}{h(xz)}\longrightarrow_{\tr(z)\rightarrow\infty}\ 
	\left\{
	\begin{array}{ll}
	\infty,&\mathrm{if}\ x^{\top}T^{-1}x>y^{\top}T^{-1}y;\\
	0,&\mathrm{if}\ x^{\top}T^{-1}x<y^{\top}T^{-1}y.	
	\end{array}
	\right.
	$$
	\item[(ii)] For target distribution $\tilde{\Pi}(\dif S)=\tilde{\pi}(S)\mu(\dif S)$ we will call $\tilde{\pi}$ rapidly varying when 
	$\pi(x)=\tilde{\pi}(x^{\top}U^{-1}x)$ is rapidly varying.
	\end{itemize}
\end{defi}
\noindent 
The above is a natural extension of rapid variation for scalar-valued functions; see, e.g., Section 2.4 of \cite{opac-b1134644} for details on the scalar case.
Rapid variation is relevant for several light-tailed distributions. The Lebesgue density of Matrix-Normal distribution  is rapidly varying, as is the density $\tilde{\pi}=\tilde{\pi}(S)$  
of the Wishart distribution in (\ref{eq:WD}).


\begin{prop}
\label{prop:drift}
If $\tilde{\pi}(S)$ is strictly positive, continuous, rapidly varying function, then (\ref{eq:drift1}) holds for 
drift function $\mathsf{V}(S)=\tilde{\pi}(S)^{-\alpha}$, with any $\alpha\in(0,1)$. 
\end{prop}
\begin{proof}
The proof is given in Appendix \ref{app:drift}.
\end{proof}
\noindent It remains to establish (\ref{eq:drift2}), where
$\det(S)\rightarrow 0$, 
or equivalently, $S\rightarrow S_0$,
with $S_0$ degenerate, positive semi-definite, symmetric matrix. 
We prove the drift inequality for the special case $S_0=0$; 
this also provides a proof of (\ref{eq:drift2}) in the trivial scalar case $q=1$.

\begin{prop}
Assume that $\tilde{\pi}$ is strictly positive, continuous. 
Suppose that for some $\zeta\neq 0$, for any $\epsilon\in P^+(q)$, 
\[
\frac{\tilde{\pi}(\epsilon\circ S)}{\tilde{\pi}(S)}\longrightarrow_{\tr (S)\rightarrow 0} \det(\epsilon)^{\zeta/2}. 
\]
For $\mathsf{V}(S)=\tilde{\pi}(S)^{-\alpha}$, $\alpha>0$, the MpCN kernel $\tilde{P}$ on $P^{+}(q)$, with target $\tilde{\Pi}$, satisfies,
\[
\limsup_{\tr (S)\rightarrow 0}\frac{\tilde{P}\mathsf{V}(S)-\mathsf{V}(S)}{\mathsf{V}(S)}<0. 
\]
\end{prop}

\begin{proof}
Recall that the proposal writes as $\epsilon\circ S$ with $\epsilon\sim\mathcal{L}(\dif\epsilon)$. 
By the dominated convergence theorem,
\begin{align*}
\frac{\tilde{P}\mathsf{V}(S)-\mathsf{V}(S)}{\mathsf{V}(S)}
&=\mathbb{E}\,\Big[\,\Big\{\Big(\frac{\tilde{\pi}(\epsilon\circ S)}{\tilde{\pi}(S)}\Big)^{-\alpha}-1\Big\}\min\Big\{1,\frac{\tilde{\pi}(\epsilon\circ S)}{\tilde{\pi}(S)}\Big\}\,\Big]\\
&\qquad \rightarrow 
\mathbb{E}\,\big[\,\big\{\det(\epsilon)^{-\alpha \zeta/2}-1\big\}\min\big\{1,\det(\epsilon)^{\zeta/2}\big\}\,\big]. 
\end{align*}
Since $\epsilon$ and $\epsilon^{-1}$ have the same law, we obtain,
\begin{align*}
\mathbb{E}\,&\big[\,\big\{\det(\epsilon)^{-\alpha \zeta/2}-1\big\}\min\big\{1,\det(\epsilon)^{\zeta/2}\big\}\,\big]\\
&=
\mathbb{E}\,\big[\,\big\{\det(\epsilon)^{-\alpha \zeta/2}-1\big\},\det\epsilon\ge 1\,\big]+
\mathbb{E}\,\big[\,\big\{\det(\epsilon)^{-\alpha \zeta/2}-1\big\}\det(\epsilon)^{\zeta/2},\det\epsilon< 1\,\big]\\
&=
\mathbb{E}\,\big[\,\big\{\det(\epsilon)^{-\alpha \zeta/2}-1\big\},\det\epsilon\ge 1\,\big]+
\mathbb{E}\,\big[\,\big\{\det(\epsilon)^{\alpha \zeta/2}-1\big\}\det(\epsilon)^{-\zeta/2},\det\epsilon> 1\,\big]\\
&=
\mathbb{E}\,\big[\,\big\{\det(\epsilon)^{-\alpha \zeta/2}-1\big\}\big\{1-\det(\epsilon)^{-(1-\alpha)\zeta/2}\big\},\det\epsilon\ge 1\,\big],
\end{align*}
with the last expression being negative. 
\end{proof}
\noindent 
%
\revision{In the proof above, the negativity of the limiting integral follows directly from the symmetry $\mathcal{L}(\epsilon)=\mathcal{L}(\epsilon^{-1})$ when $S\rightarrow 0$. The general scenario $S\rightarrow S_0$ is more complicated. In this case, the limiting integral depends on $S_0$, and the symmetry does not simplify the integral.  Thus, the above proof cannot be applied to the case $S_0\neq 0$.  }

 To stress the effect of algorithmic parameter $\rho$, we write the MpCN kernel on $M(p,q)$ (resp.~$P^{+}(q)$) as $P_{\rho}$ (resp.~$\tilde{P}_{\rho}$) and the corresponding proposal 
as $\tilde{Q}_\rho$. In all cases, $\rho\in [0,1)$. 

\begin{prop}
\label{prop:spetral_gap_quivalence}
\begin{itemize}
\item[(i)] 
If MpCN kernel $P_0$ on $M(p,q)$ has a spectral gap, then so does $P_\rho$ for any $\rho\in[0,1)$.
\item[(ii)] If MpCN kernel $\tilde{P}_0$ on $P^+(q)$ has a spectral gap, then so does $\tilde{P}_\rho$ for any $\rho\in[0,1)$.
\end{itemize}
\end{prop}
\begin{proof}
The proof is given in Appendix \ref{app:rho}.
\end{proof}

\begin{rem}
See Appendix \ref{app:rho} for details on the concept of spectral gap.
By Proposition~\ref{prop:spetral_gap_quivalence}, geometric ergodicity of the MpCN kernel 
with target $\Pi$ on $M(p,q)$ (or $\tilde{\Pi}$ on $P^{+}(q)$) and parameter $\rho\in[0,1)$
is implied by that of the MpCN kernel for $\rho=0$.
The result is important, as working with $\tilde{P}_0$ simplifies a lot the involved matrix calculations for deriving drift conditions. 
From a practical point of view, the result allows for numerical evidence over inequality (\ref{eq:drift2}) -- as in Section \ref{sec:drift2} that follows -- by taking advantage of the fact 
that the choice $\rho=0$ provides a much more manageable expression for the distribution of the noise $\epsilon$ in
(\ref{eq:gamma}), involved in the MpCN proposal (Proposition \ref{prop:eigen} below exploits this fact).
%
\end{rem}

\begin{prop}
\label{prop:eigen}
If $\rho=0$, the eigenvalues $\lambda_1\le \lambda_2\le\cdots\le \lambda_q$ of $\epsilon\sim\mathcal{L}(\dif\epsilon)$ have the following joint Lebesgue density function, up to a normalising constant,
\begin{align*}
    \prod_{i=1}^q\frac{\lambda_i^{(p-q-1)/2}}{(1+\lambda_i)^{p}}\prod_{i<j}(\lambda_j-\lambda_i). 
\end{align*}
\end{prop}

\begin{proof}
This follows from Theorem 5.3.1 of \cite{MR770934} together with the analytical expression of $\mathcal{L}(\dif\epsilon)$, since, for any integrable $f$,
\begin{align*}
    \int f(\epsilon)\mathcal{L}(\dif \epsilon)~&\propto\int f(\epsilon) \frac{\det(\epsilon)^{p/2}}{
	\det(I_q+\epsilon)^p}
	\mu(\dif \epsilon)\\
&=\int f(\epsilon) \frac{\det(\epsilon)^{(p-q-1)/2}}{
	\det(I_q+\epsilon)^p}
	\operatorname{Leb}^{q(q+1)/2}(\dif\epsilon)\\
&= \int \prod_{i=1}^q\frac{\lambda_i^{(p-q-1)/2}}{
	(1+\lambda_i)^p}\prod_{i<j}(\lambda_j-\lambda_i)\dif\lambda_1\cdots\dif\lambda_q
	 \int_{\mathcal{O}(q)}f(\mathsf{U}\Lambda\mathsf{U}^{\top})\dif\mathsf{U},
\end{align*}
where $\Lambda$ is diagonal with elements $\lambda_1,\ldots, \lambda_q$. 
\end{proof}
\noindent 

\subsection{Numerical Evidence over Drift Condition (\ref{eq:drift2})}
\label{sec:drift2}

We consider the case of a Wishart target $\tilde{\Pi}=W_q(r, I_q)$ with $r\in\RR$, $r>q-1$ for $\mathsf{V}(S)=\tilde{\pi}(S)^{-\alpha}$. In this case, 
\begin{align}\label{eq:drift_at_singular}
    \frac{\tilde{P}_0\mathsf{V}(S)-\mathsf{V}(S)}{\mathsf{V}(S)}=\int_{P^+(q)} \{\eta(\epsilon, S)^{-\alpha}-1\}\min\{1,\eta(\epsilon, S)\}\mathcal{L}(\dif \epsilon),
\end{align}
where, 
\[
\eta(\epsilon, S)=\frac{\tilde{\pi}(\epsilon\circ S)}{\tilde{\pi}(S)}
=(\det\epsilon)^{r/2}
\exp\big(-\tfrac{1}{2}\tr\,[\,(\epsilon-I_q)S\,]\,\big). 
\]
(Note that the right-hand side of (\ref{eq:drift_at_singular}) can be defined even if $S$ is degenerate.) We want to numerically investigate (\ref{eq:drift2}).
By continuity, it suffices to show that (\ref{eq:drift_at_singular}) is always negative if $S$ is degenerate. Furthermore, we can assume that $S$ is diagonal. To see this, first observe that for any $S\in P^+(q)$, there is $\mathsf{U}\in\mathcal{O}(q)$ such that $\mathsf{U}^{\top}S\mathsf{U}$ is diagonal.  We also have $\eta(\mathsf{U}\,\epsilon\,\mathsf{U}^{\top}, S)=\eta(\epsilon, \mathsf{U}^{\top}S\mathsf{U})$, and 
the law of $\mathsf{U}\,\epsilon\,\mathsf{U}^{\top}$ is the same as that of $\epsilon$ when $\rho =0$. Therefore, it is enough to show that (\ref{eq:drift_at_singular}) is always negative if $S$ is a degenerate diagonal matrix. 
The law of the eigenvalues $\lambda_1,\ldots, \lambda_q$ is determined in Proposition \ref{prop:eigen}, and $\epsilon$ is decomposed as $\epsilon=\mathsf{U}\Lambda\mathsf{U}^{\top}$ where $\mathsf{U}$ is uniformly distributed in $\mathcal{O}(q)$ and $\Lambda$ is a diagonal matrix with diagonal elements  $\lambda_1,\ldots,\lambda_q$. 
Thus, we can now evaluate (\ref{eq:drift_at_singular}) via numerical integration. 
\revision{To evaluate this integral, we used  importance sampling with the Pareto distribution as a reference measure. In Fig.~\ref{fig:a2} we use $10^5$ random samples for each given of $S$.  }


We fix $q=2$, so we can assume that $S=\mathrm{diag}(s,0)$, $s>0$. As we see in Fig.~\ref{fig:a2}, left panel, for small enough $\alpha\in (0,1)$, the value of (\ref{eq:drift_at_singular}) is (numerically found to be) negative. Thus, we can numerically confirm the drift condition for $\tilde{\Pi}=W_q(r, I_q)$ for  $r=p=2$ from Fig.~\ref{fig:a2}, left panel, and for $r=4$, $p=5$ from Fig.~\ref{fig:a2} right \revision{(recall the law of $\epsilon$ depends on $p$)}. Other choices of $r$ and $p$ yield similar figures. A Dirichlet form argument can be used to 
show that geometric ergodicity for $\tilde{\Pi}=W_q(r, I_q)$ implies geometric ergodicity for $\tilde{\Pi}=W_q(r, T)$, 
for general $T\in P^+(q)$. 
Recall also --  see the comment after Proposition \ref{prop:inv} 
-- that geometric ergodicity of the MpCN kernel for a Wishart target implies geometric ergodicity  for an Inverse-Wishart target.

\begin{figure}[!ht]
\begin{center}
\includegraphics[width=0.49\columnwidth]{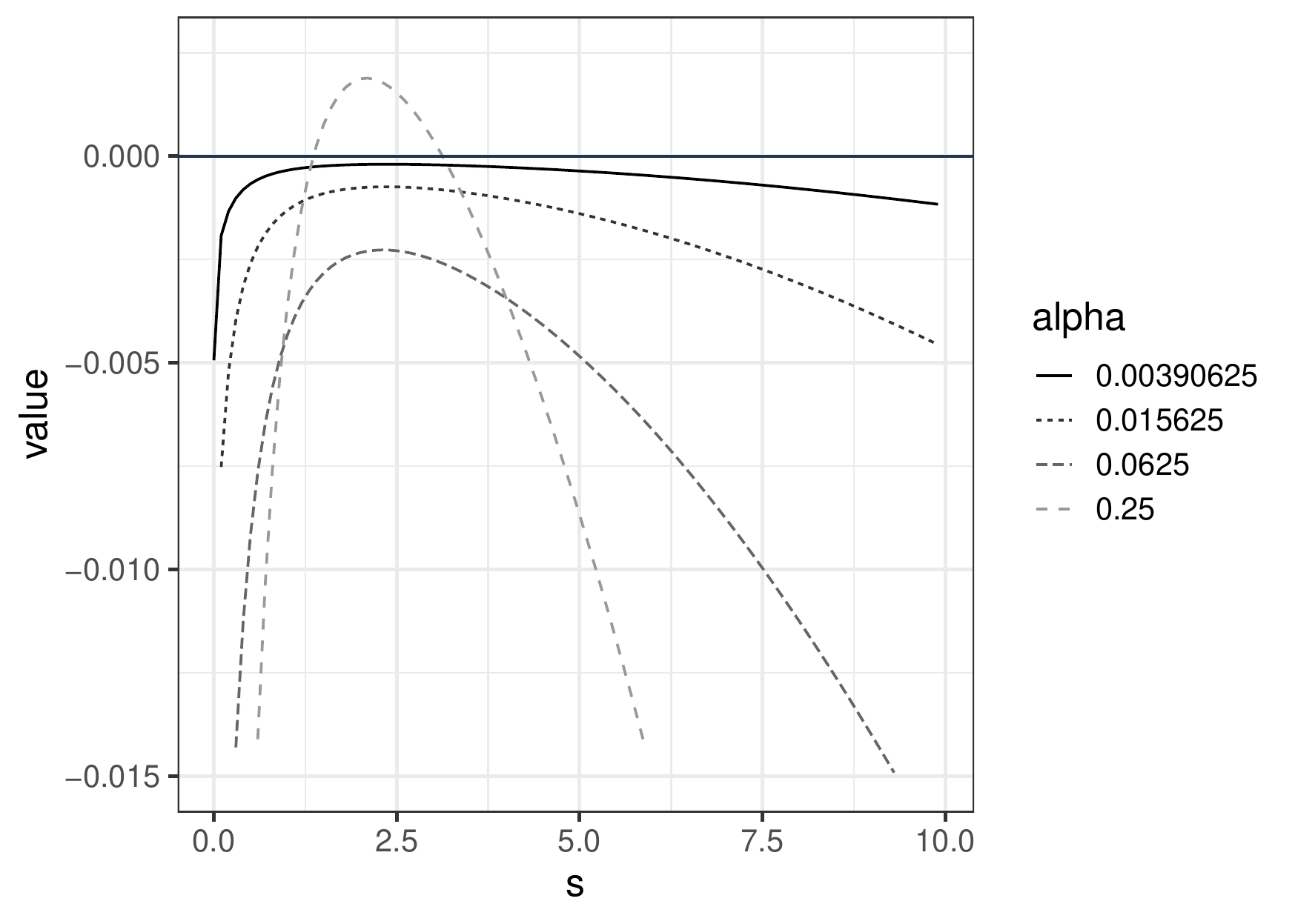}
\includegraphics[width=0.49\columnwidth]{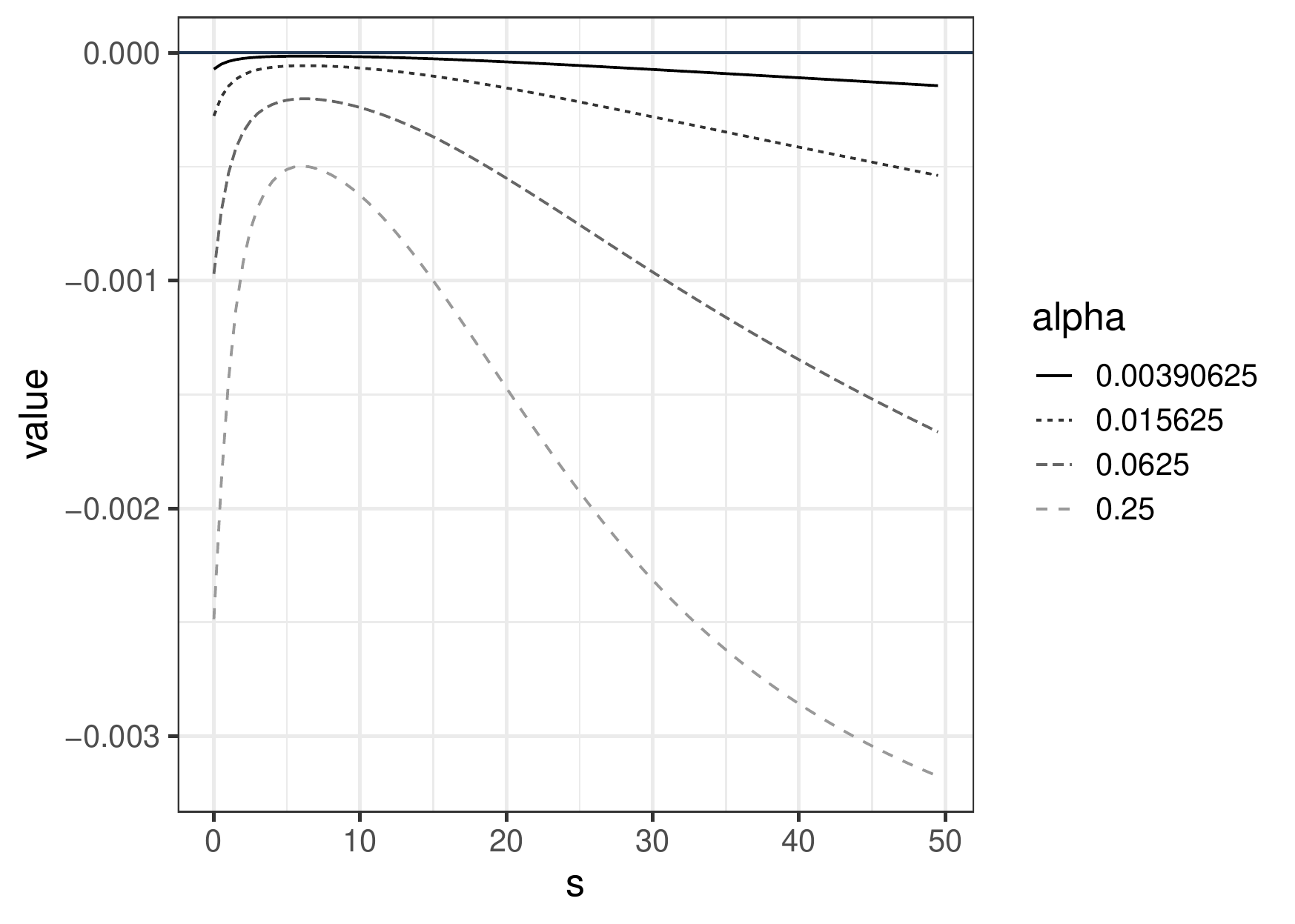}
\vspace{-0.3cm}
\caption{Numerical evaluation of the relative expected drift change  $(\tilde{P}_0\mathsf{V}(S)-\mathsf{V}(S))/\mathsf{V}(S)$
for target $\tilde{\Pi}=W_q(r, I_q)$,
at $S=\mathrm{diag}(s,0)$, for $r=p=2$ (left) and   $r=4$, $p=5$ (right).}
\label{fig:a2}
\end{center}
\vspace{-0.2cm}
\end{figure}


\section{Simulation Experiments}
\label{sec:results}

In this section, we discuss and analyse the difference in performance  among the algorithms considered in this work, i.e.~RWM, pCN and MpCN. 

\subsection{Tuning Parameters and Performance on Simple Targets}
First, the choice of tuning parameters  is discussed. Consider a target distribution defined on $M(p,q)$. 
Matrices $U$ and $V$, appearing in RWM and pCN, are scaling parameters that can be learned from a initial phase of the MCMC algorithms. Upon recalling the second moment properties of the Matrix-Normal law in Remark \ref{rem:connect}(i), we can have the estimates, 
\begin{align*}
\hat{U} = \hat{E}_1, \quad \hat{V} = \hat{E}_2/\tr(\hat{E}_2),
\end{align*}
having defined,
\begin{align*}
\hat{E}_1=(L-1)^{-1}\sum_{l=1}^L (X_l-\bar{X})(X_l-\bar{X})^{\top},\quad 
\hat{E}_2=(L-1)^{-1}\sum_{l=1}^L (X_l-\bar{X})^{\top}(X_l-\bar{X}), 
\end{align*}
where $X_1,..., X_L$, $L\ge 1$, are samples from the MCMC method and $\bar{X}$ is the sample mean.
In the case of RMW, one should introduce a scalar, $\sigma^2>0$, to allow for controlling the acceptance probability, so that the proposal writes as $y = x + \sigma\,\xi$, $\xi\sim N(0, \hat{U}, \hat{V})$ -- this is not needed in the case of pCN and MpCN.
%
%
%
Recall that MpCN does not involve $V$ as above, but simulates from an Inverse-Wishart law. 
As $\rho$ decreases and $\sigma^2$ increases, the acceptance probability decreases.
Empirically, a good choice of average acceptance probability for all algorithms is around $20\%$ to $40\%$.

Assume that the target distribution is defined on $P^+(q)$, and that $\U=U$, in which case the MCMC algorithms give rise to Markov chains on $P^+(q)$ with dynamics that do not depend on $U$.
Without loss of generality, let $U=I_p$. 
The dimension $p$ is a tuning parameter.
Let $S\in P^+(q)$ and $S=x_p^{\top}x_p$ for some $x_p\in M(p,q)$.  For RWM, the proposed value from $x_p$ is $y_p=x_p+ w_p$, $w_p\sim N_{p,q}(0, I_p,\sigma^2~V)$. 
We set $S^*= y_p^{\top}y_p$. Then $\sigma^{-2}S^*$ follows the noncentral Wishart distribution with noncentral matrix $\sigma^{-2}S$, covariance $V$, and $p$ degrees of freedom \citep[see Section 1.5.4 of][]{MR1960435}. By the properties of non-central Wishart distribution, $S^*$ has the same law as that of the matrix sum, 
$$
u^{\top}u+v^{\top}v=\begin{pmatrix}u\\v\end{pmatrix}^{\top}\begin{pmatrix}u\\v\end{pmatrix},\quad u\sim N_{q,q}(S^{1/2}, I_q, \sigma^2~V),\quad  v\sim N_{p-q,q}(0, I_{p-q}, \sigma^2~V). 
$$
The first term on the left side does not depend on $p$, while the second term is 
$\sum_{i=1}^{p-q} v_iv_i^{\top}$, where $v_i\sim N_q(0, V)$, and it is monotonically increasing with $p$. Thus $S^*$ increases monotonically with $p$, an effect that is illustrated 
in Fig.~\ref{fig:param1}, top-left panel. In Fig.~\ref{fig:param1}, the $x$-axis and the $y$-axis are the logarithms of 
two out of $q$ eigenvalues 
of the proposed position $S^*$, when the current position $S$ is the identity. A similar calculation shows that the above observation is also true for pCN (Fig.~\ref{fig:param1}, top-centre panel). In contrast, for MpCN the behavior of $S$ for varying $p$ is different. Recall that the matrix $V$ is random in this case. By the law of large numbers, as $p\rightarrow \infty$, 
$$
\frac{V^{-1}}{p}\longrightarrow (x^{\top}x)^{-1}=S^{-1},\quad 
(w^{\top}w)\longrightarrow S,
$$
where $V$ and $w$ are as in (\ref{eq:mpcn_proposal}). 
Also, $x^{\top}w\sim N_{q,q}(0, S, V)$ converges to $0$ since $V\rightarrow 0$. One thus has
$(y^{\top}y)\longrightarrow (x^{\top}x)$, and the acceptance probability converges to $1$, where $y$ is the proposed value of $x$. Therefore, the kernel degenerates as $p\rightarrow \infty$ (see Fig.~\ref{fig:param1}, top-right panel). 
In general, in addition to the choice of $\sigma$ and $\rho$ (Fig.~\ref{fig:param1}, bottom panel), the choice of $p$ provides additional flexibility in the design of the MCMC kernels. 

\begin{figure}[!ht]
\begin{center}
\includegraphics[width=0.3\columnwidth]{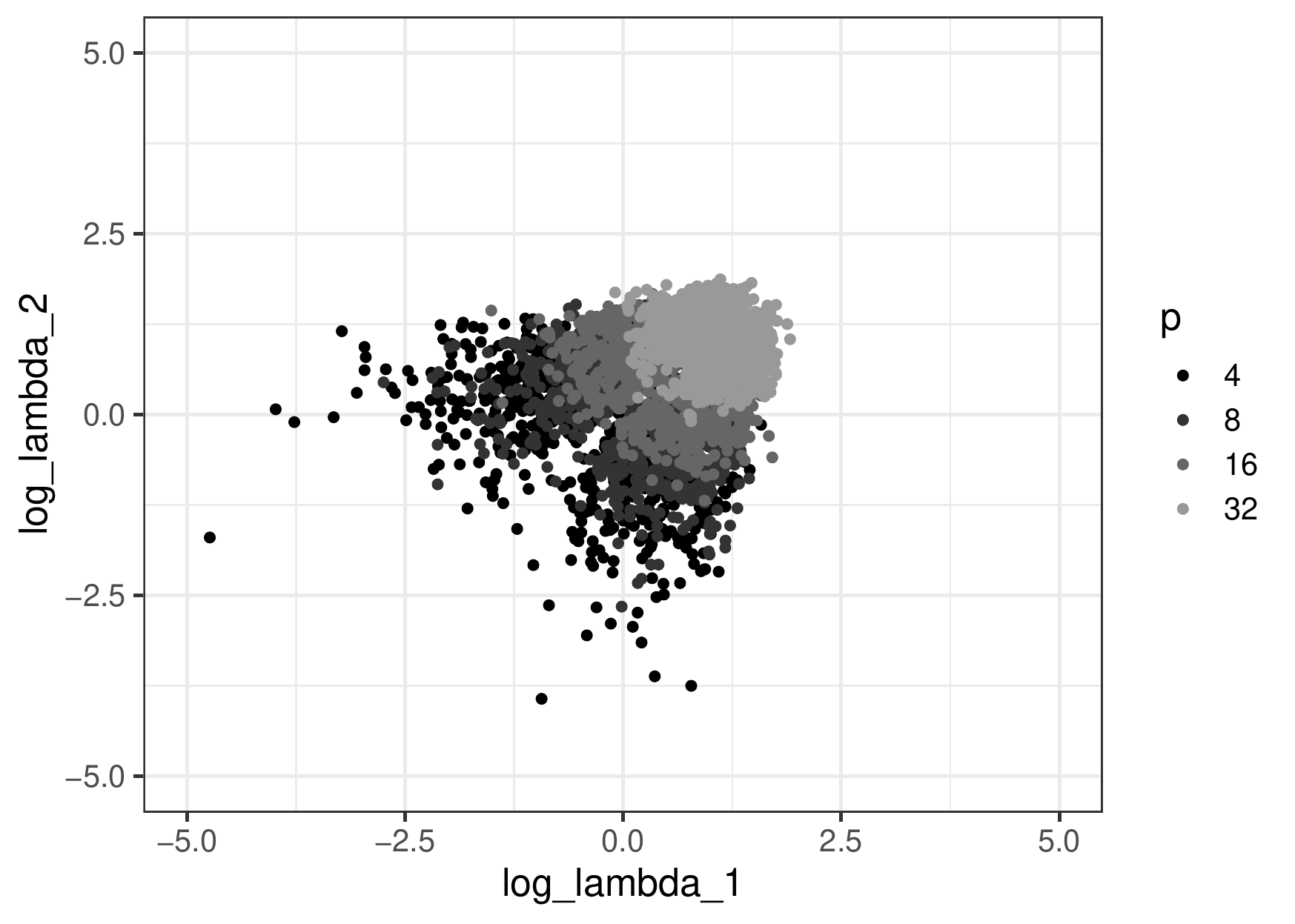}
\includegraphics[width=0.3\columnwidth]{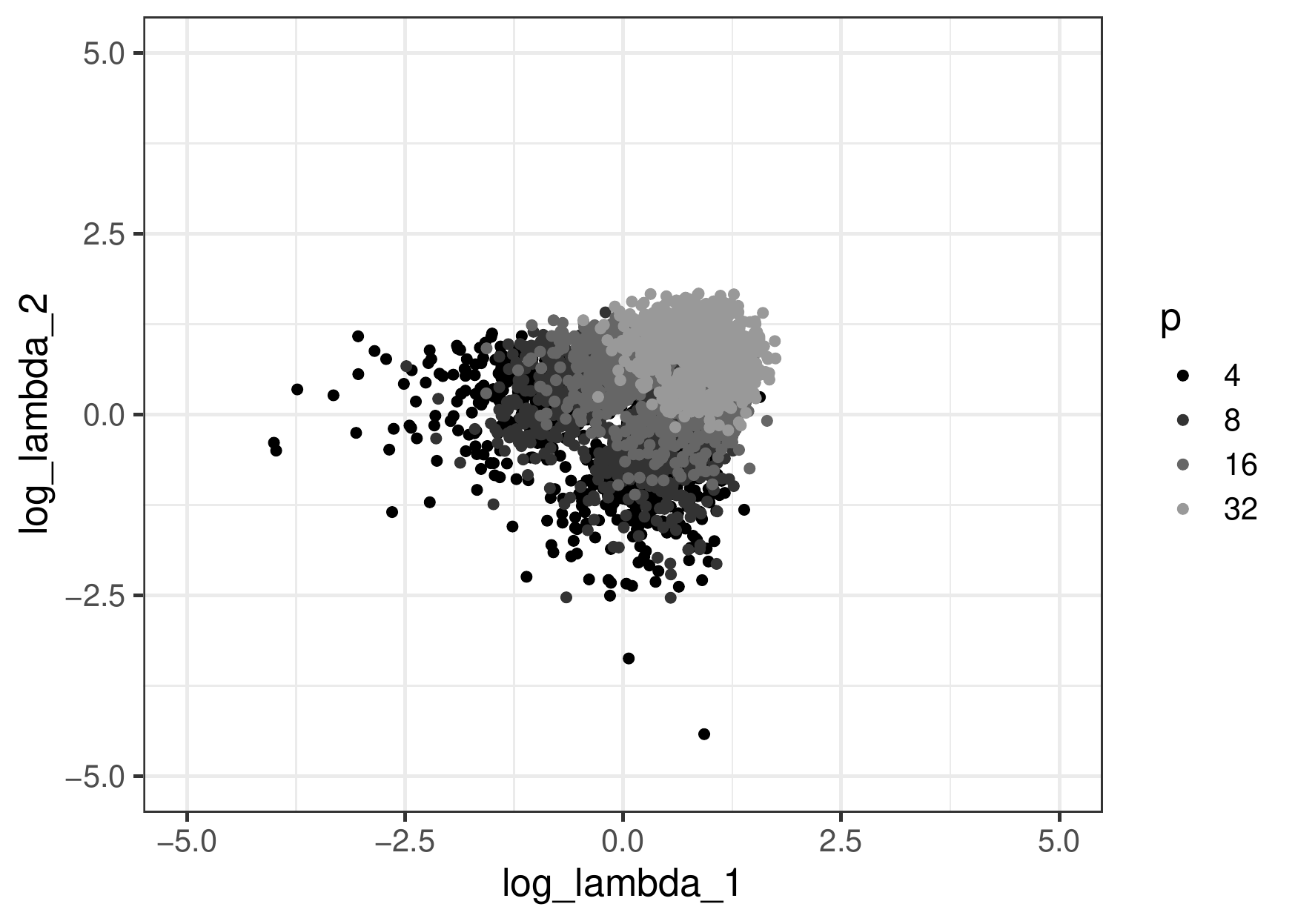}
\includegraphics[width=0.3\columnwidth]{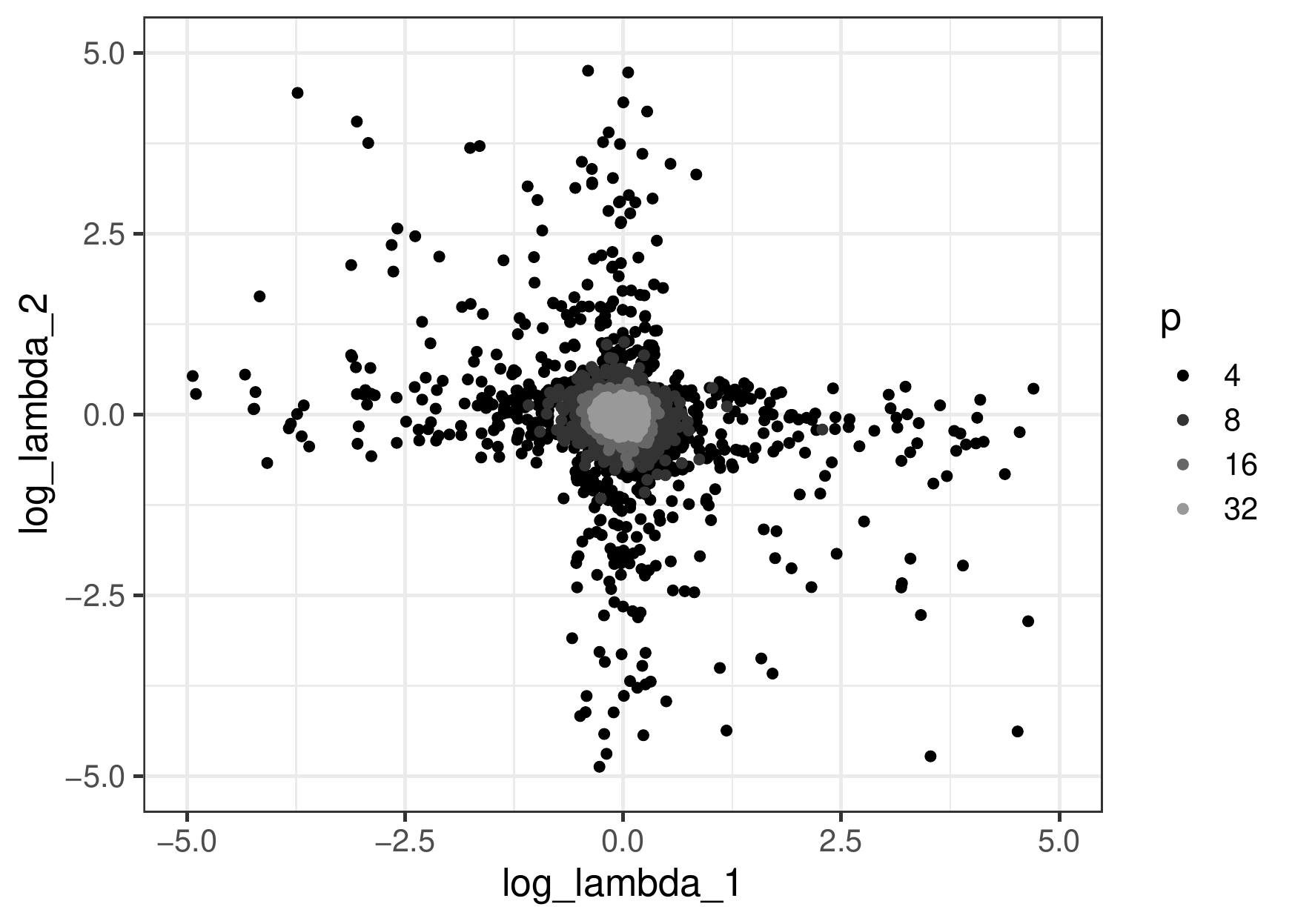}\\
\includegraphics[width=0.3\columnwidth]{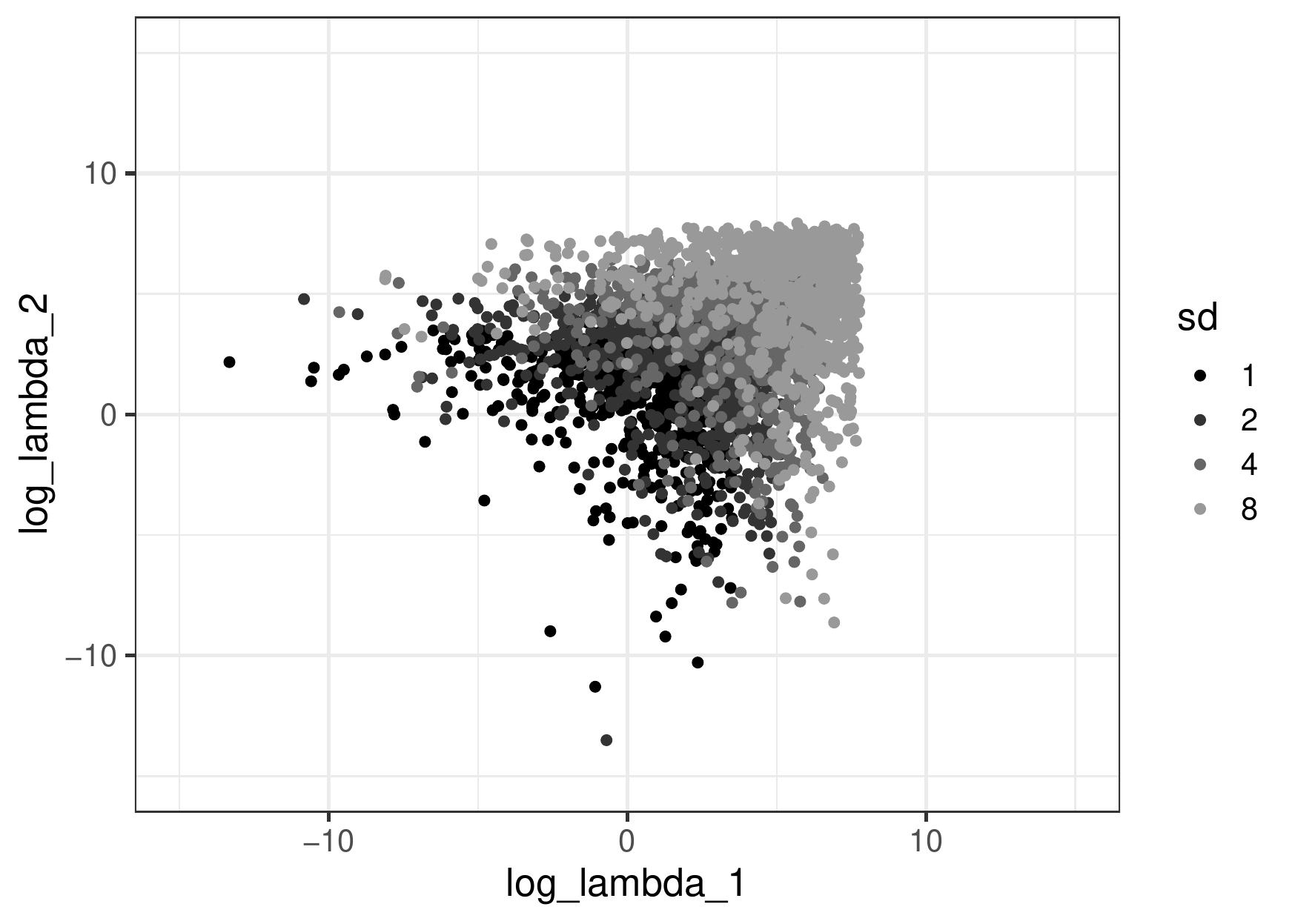}
\includegraphics[width=0.3\columnwidth]{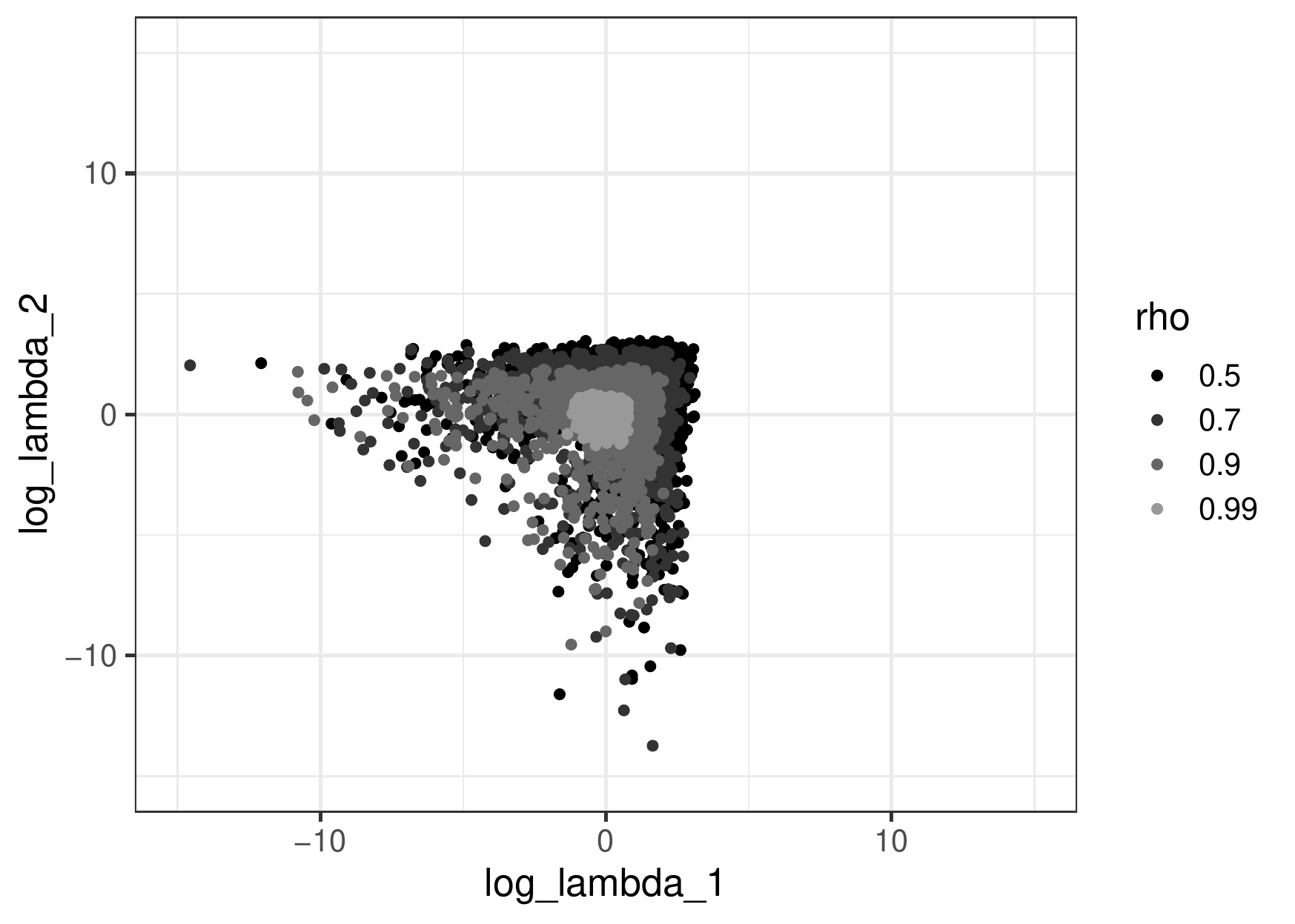}
\includegraphics[width=0.3\columnwidth]{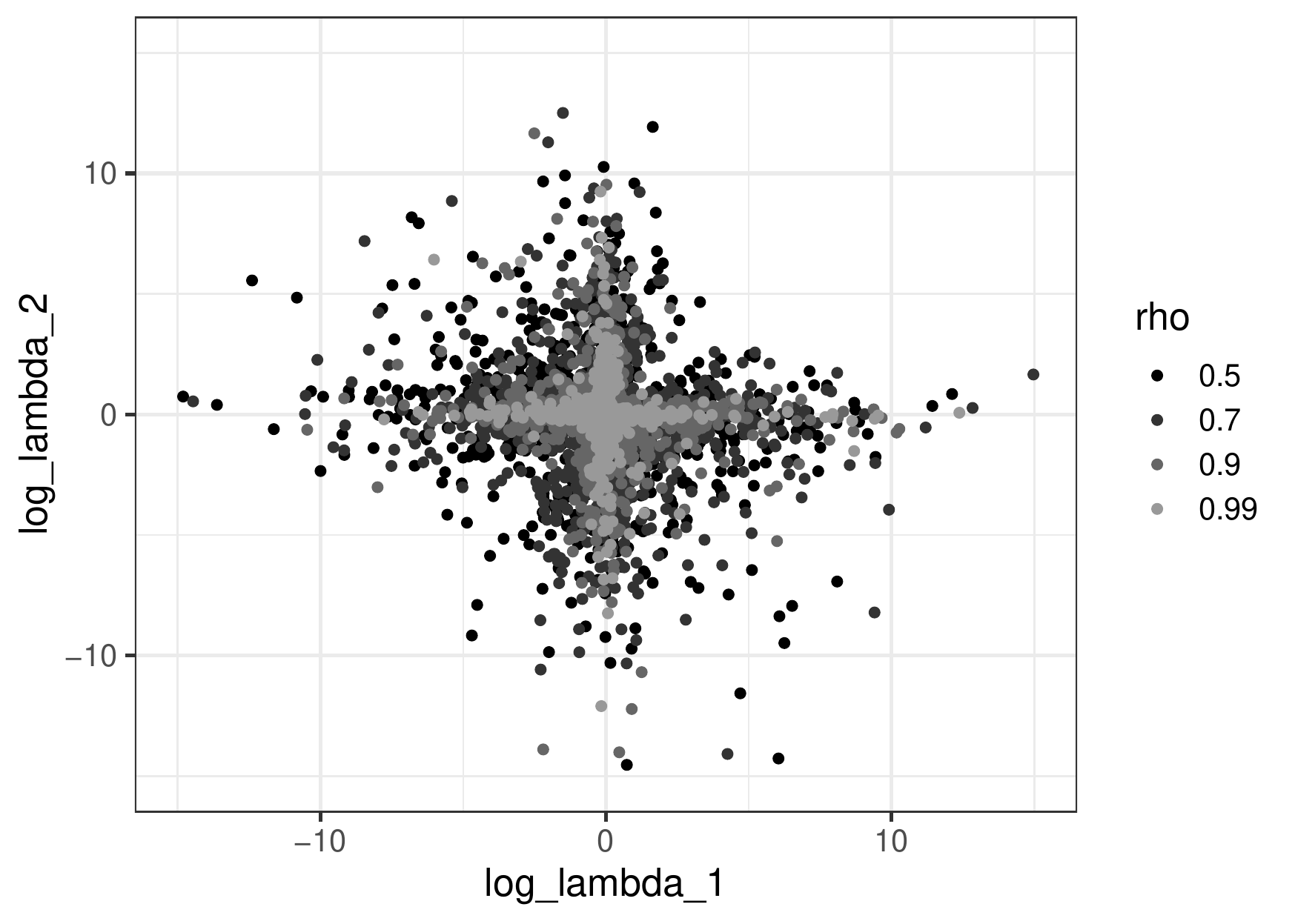}
\vspace{-0.4cm}
\caption{The effect of increasing $p$ (top) and the effect of $\sigma$ and $\rho$ (bottom) for the 
proposed position $S^{*}$ when the current position is $S=I_q$, for RWM (left), pCN  (center) and MpCN (right).  
In all cases, $q=4$, for RWM and pCN we have $V=I_q$, and $10^3$ realisations of $S^{*}$ are used for each figure.}
\label{fig:param1}
\end{center}
\vspace{-0.2cm}
\end{figure}


We now check the performance of the three algorithms for a Wishart target, $W_q(r, T)$, and an Inverse-Wishart one, $W_q^{-1}(r, T)$.  We set $q=8$, $r=16$, and produce $T$ via a sample $T\sim W_q(q, I_q)$.
In the following experiment, we set $V=I_q$ for RWM, pCN, and set $p=q$. Parameters $\sigma$, $\rho$ are chosen so that the acceptance probability is between $20\%$ and $40\%$.  We measure the distance between the empirical means of the MCMC algorithms and the mean of the target distribution via metric $\mathsf{d}(\cdot,\cdot)$ defined in (\ref{eq:metric}).
The computational cost for each iteration of MpCN is two to three times that of RWM and pCN \revision{since MpCN uses a simulation of the inverse Wishart distribution, that requires the evaluation of an inverse matrix and the eigencomposition of a $q\times q$ matrix per iteration}. Even taking this into account, MpCN performs much better for both the Wishart and 
Inverse-Wishart target (Fig.~\ref{fig:comparison}). 
For the Wishart target, pCN performs much worse than the other two methods, whereas pCN and RWM show similar weak performance for the Inverse-Wishart target. These findings align with the theory since pCN is not geometrically ergodic for the Wishart target, whereas both pCN and RWM are not geometrically ergodic for the Inverse-Wishart. For the heavy-tail scenario, the choice of an incremental distribution with heavy tails sometimes improves the performance of RWM \citep{MR1996270,MR2396939}. However, attempting some choices in this direction did not improve the results in our setting, thus numerics from such methods have not been included in the above plots. 

\begin{figure}[!ht]
\begin{center}
\includegraphics[width=0.4\columnwidth]{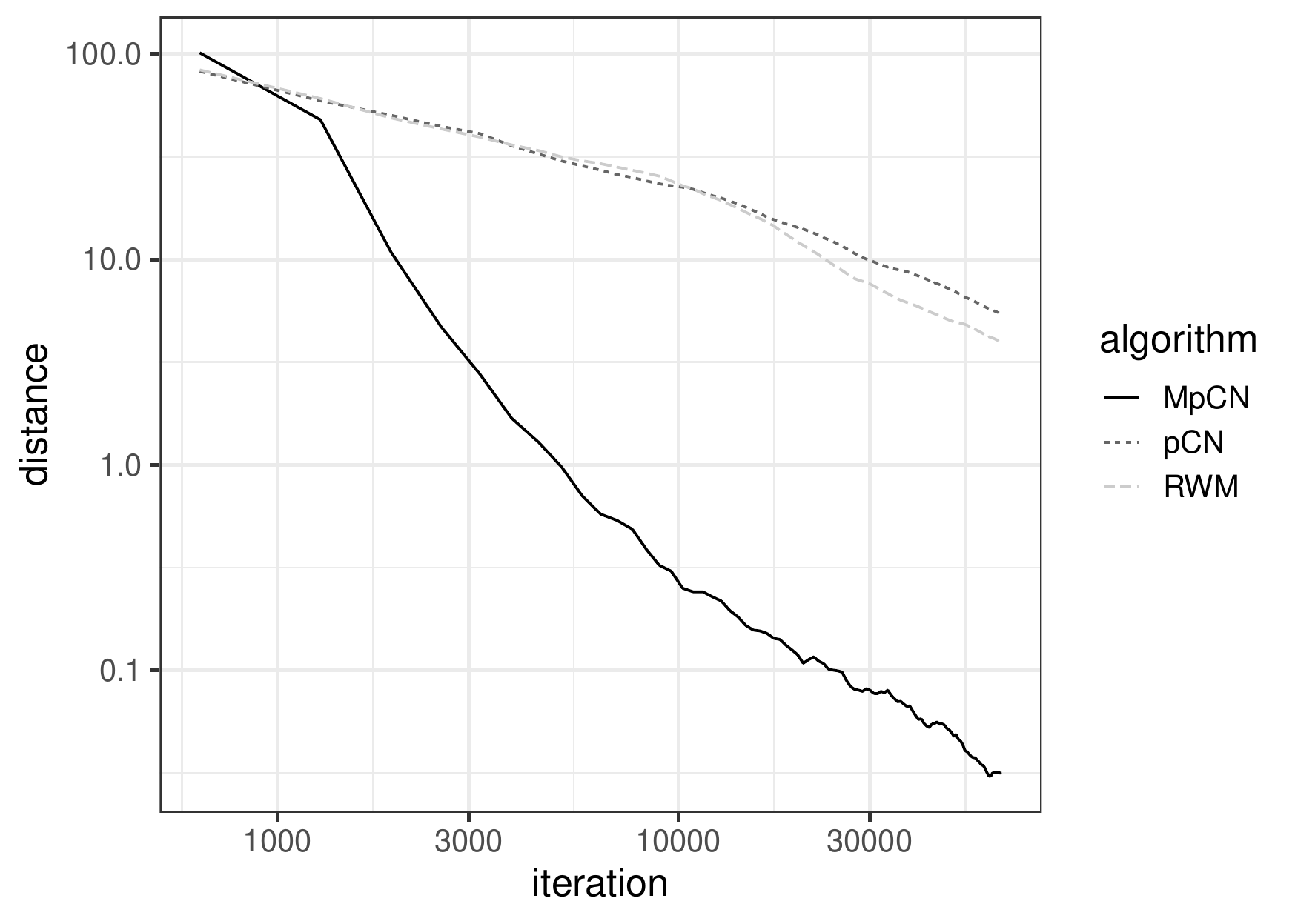}
\includegraphics[width=0.4\columnwidth]{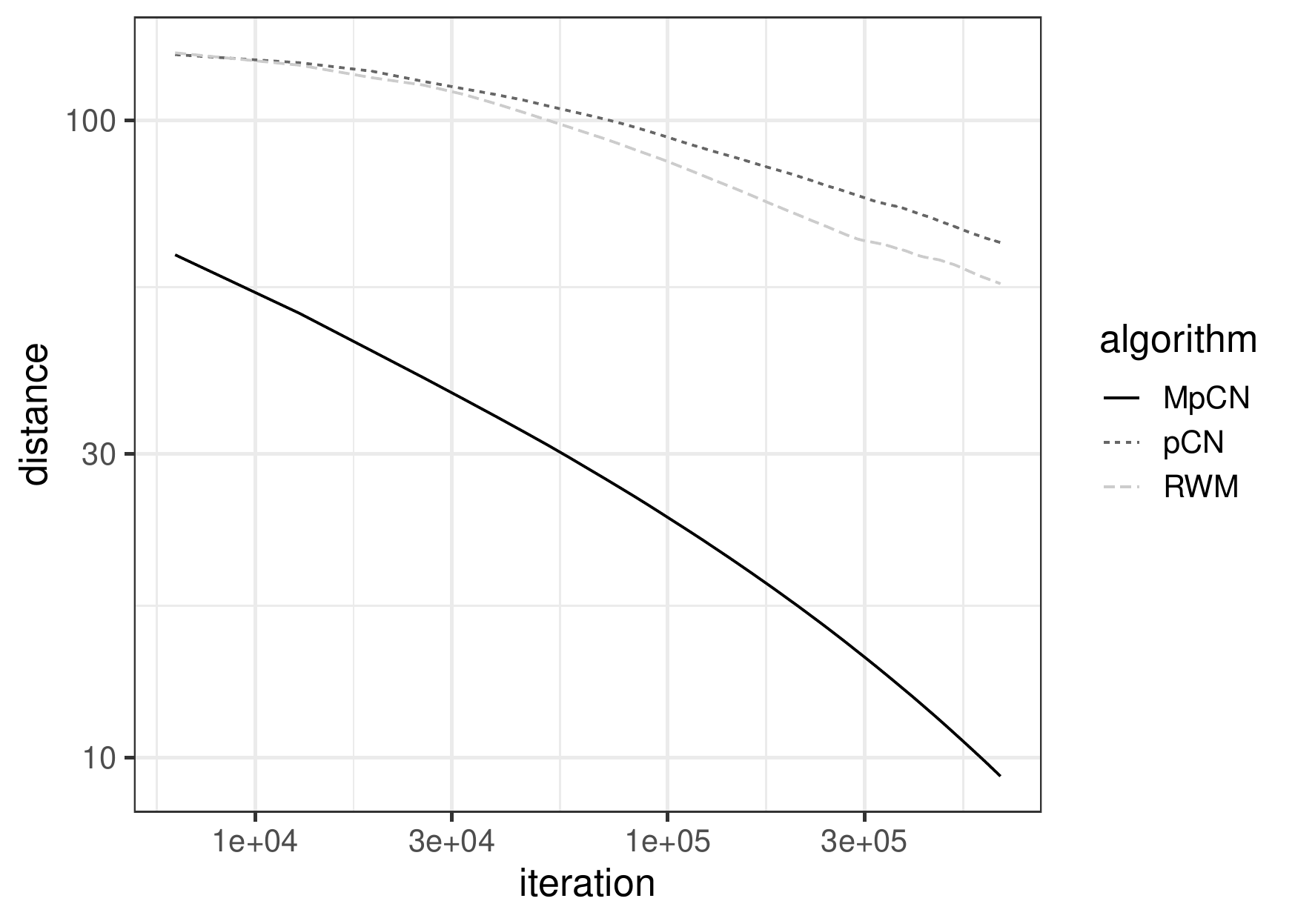}

\vspace{-0.4cm}
\caption{The distance between the empirical average and the true mean for a Wishart target (left) and an Inverse-Wishart target (right) in \revision{$\log-\log$ scale}. 
}
\label{fig:comparison}
\end{center}
\vspace{-0.2cm}
\end{figure}

\subsection{SDE Models for Covariance Matrix}
The prominent work of \cite{barn:01} introduced the scalar, non-Gaussian, Ornstein--Uhlenbeck (OU) stochastic volatility (SV) class of models of the form,
\begin{gather}
\dif\sigma^2_t = -\omega\,\sigma^2_{t^-}\dif t +
\dif \ell_t, \label{eq:OU1d}
\end{gather}
for a decay-rate parameter $\omega>0$, where the driving noise $\{\ell_t\}$ is a L\'evy process \citep{sato:1999} of positive increments and no drift -- such process is termed a `subordinator' process.
Under conditions, the differential dynamics give rise to a stationary OU process $\{\sigma^2_t\}$, with values in $\mathbb{R}_{+}$. The analysis in \cite{barn:01} illustrated that such class of models offers a great degree of flexibility in the specification of both the marginal distribution of $\sigma^2_t$ and various dynamical properties of the process, so that the model can match stylised empirical properties of observed time-series in financial economics. 
Parameter estimation in \cite{barn:01} is carried out using method of moments. 
Later, \cite{dell:15} consider the important instance in this above class of models where $\{\ell_t\}$ is a compound Poisson process -- in which case the marginal law of $\{\sigma^2_t\}$ is that of a Gamma distribution -- and develop a complex, sophisticated MCMC algorithm for carrying out full Bayesian inference. 

\cite{barn:01} briefly discuss multivariate extensions of the proposed modelling framework. This direction is explored in detail in \cite{barn:07}, where 
a non-Gaussian OU on $P^+(q)$ is carefully constructed via the differential equation,
%
\begin{gather}
\dif\Sigma_t = -(\Omega \Sigma_{t^-} + \Sigma_{t^-}\Omega^{\top})\dif t +
\dif L_t, \label{eq:matrixOU}
\end{gather}
for $\Sigma_0=\sigma_0\in P^+(q)$, $\Omega\in M(q,q)$, and matrix subordinator L\'evy process $\{L_t\}$, i.e., for $0<s<t$, $L_t-L_s$ is positive semi-definite.
%
%
\cite{barn:07} show that the conditions $\EE \log^{+}\|L_t\|_{F}<\infty$ and the 
spectrum of $\Omega$ being 
$\sigma(\Omega)\subset (0, \infty)+\mathbf{i}~\RR$, imply that SDE (\ref{eq:matrixOU}) has a stationary solution. 
The solution of (\ref{eq:matrixOU}) in-between jump times of the L\'evy process, writes as,
\begin{align*}
\Sigma_t = \exp\big\{-\Omega (t-s)\big\}\,\Sigma_s\,\exp\big\{-\Omega^{\top}(t-s)\big\}.
\end{align*}
%
%
An inferential objective arising within this important class of multivariate non-Gaussian SV models 
is the estimation of $\Omega$ and of parameters involved in the specification of $\{L_t\}$ based on
observations related to $\Sigma_t$. 

\begin{figure}[!ht]
\begin{center}
\includegraphics[width=0.4\columnwidth]{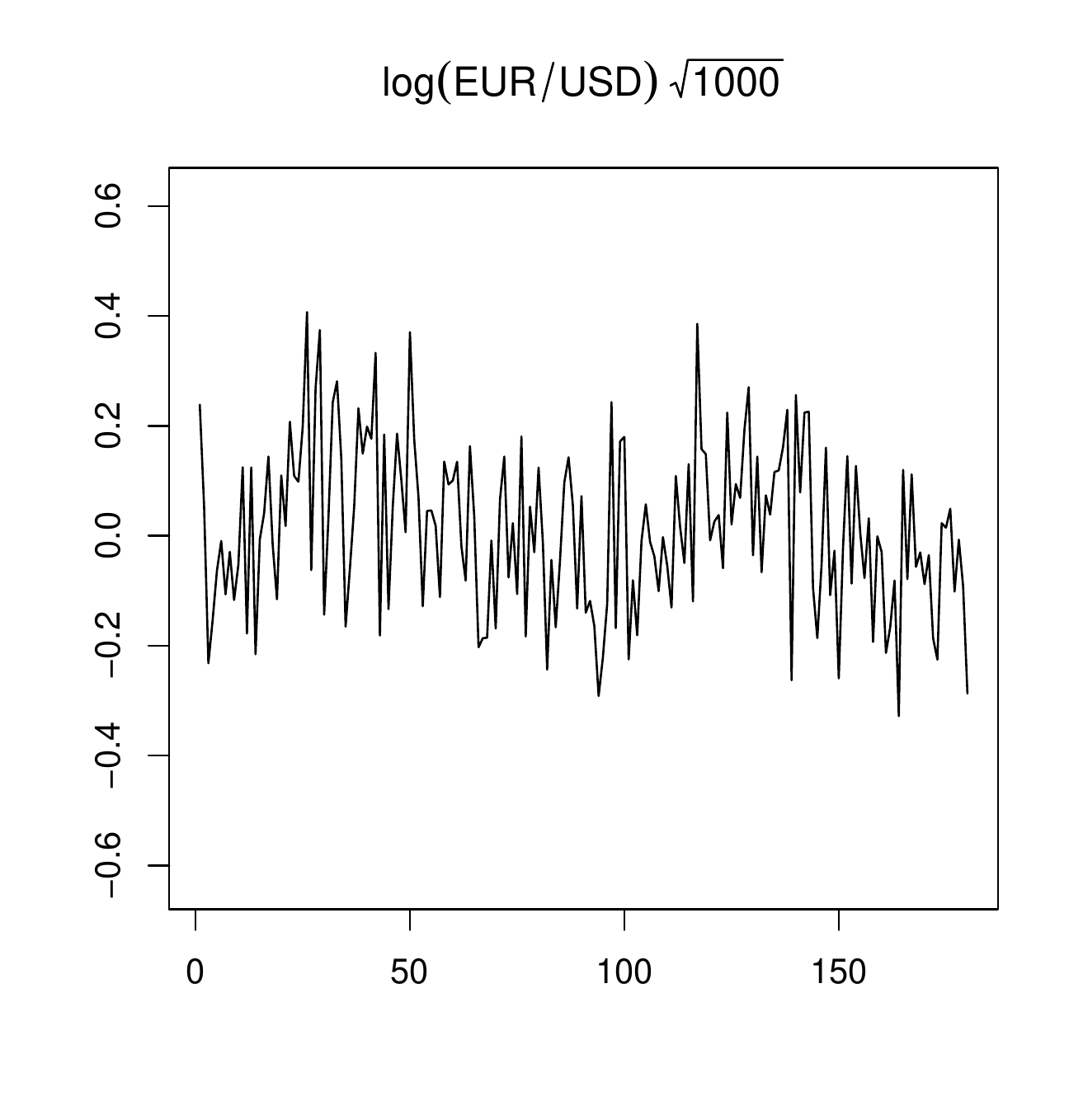}
\includegraphics[width=0.4\columnwidth]{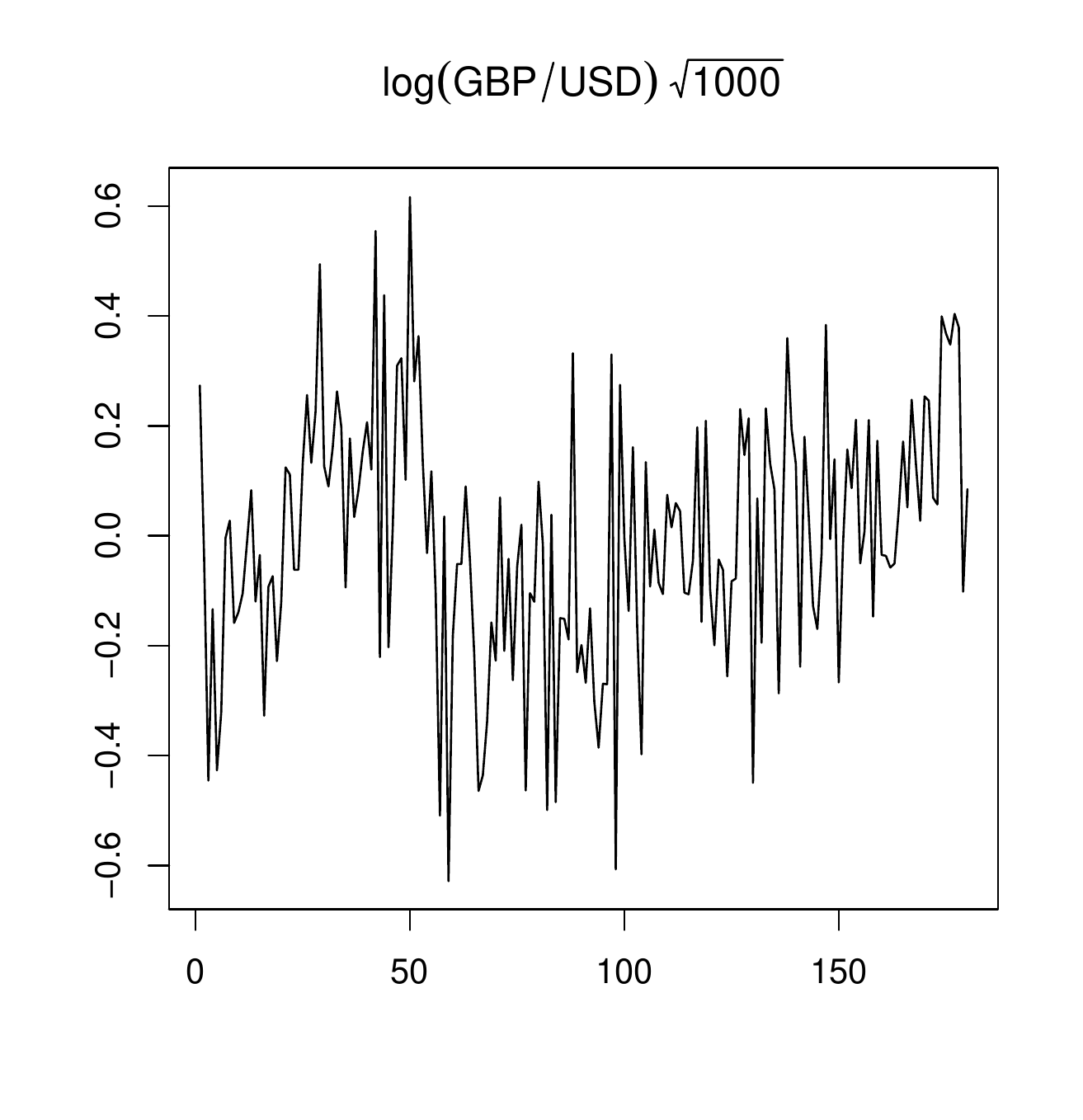}
\vspace{-0.9cm}
\caption{Observed linearly detrended $n=180$ daily log-returns -- rescaled by $\sqrt{1000}$ -- for the period 23/06/20 to 01/03/21.}
\label{fig:app1}
\end{center}
\vspace{-0.2cm}
\end{figure}

\begin{figure}[!ht]
\begin{center}
\includegraphics[width=0.49\columnwidth]{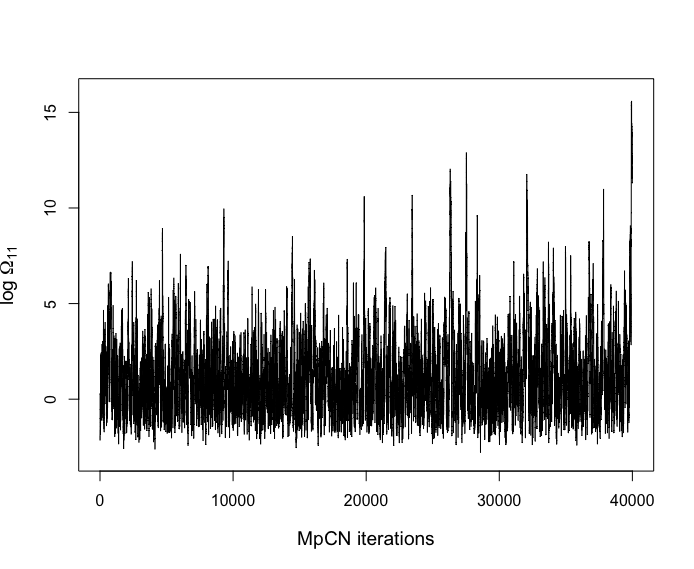}
\includegraphics[width=0.49\columnwidth]{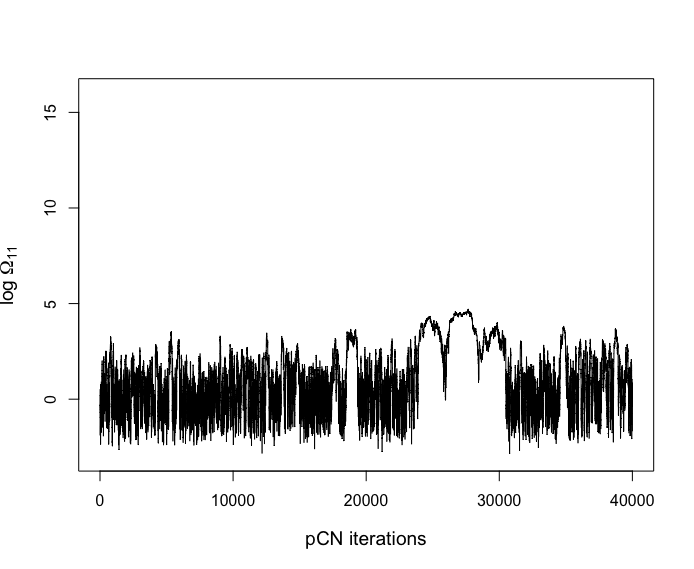}
\vspace{-0.5cm}
\caption{MpCN and pCN traceplots -- in log-scale -- for the prior of $\Omega_{11}$.}
\label{fig:app2}
\end{center}
\vspace{-0.2cm}
\end{figure}

\begin{figure}[!ht]
\begin{center}
\includegraphics[width=0.49\columnwidth]{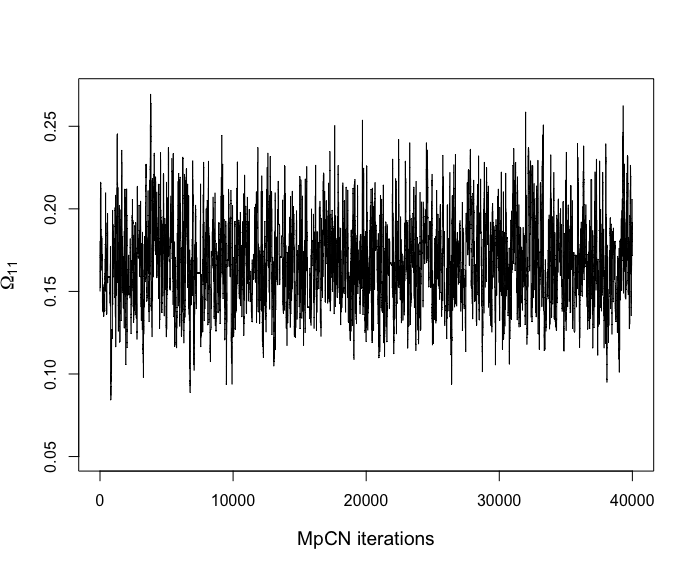}
\vspace{-0.2cm}
\includegraphics[width=0.49\columnwidth]{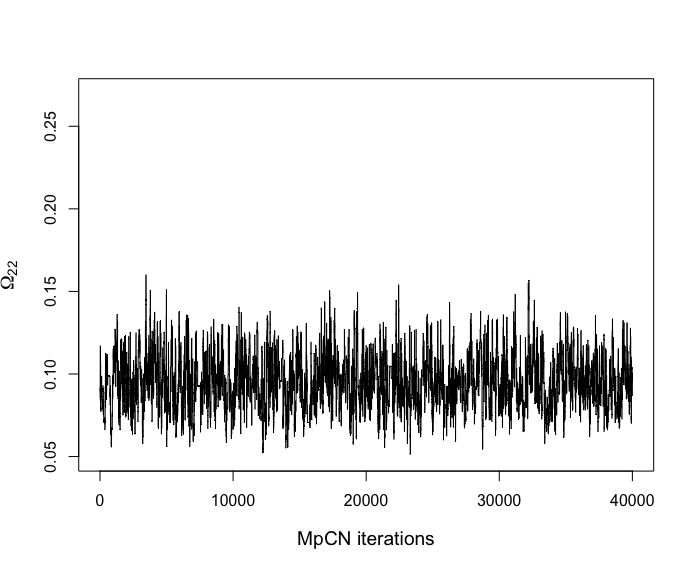}
\includegraphics[width=0.49\columnwidth]{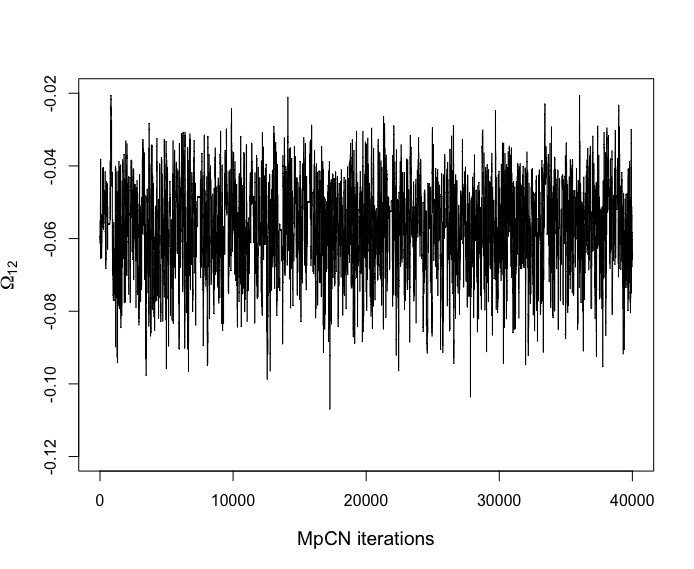}
\includegraphics[width=0.49\columnwidth]{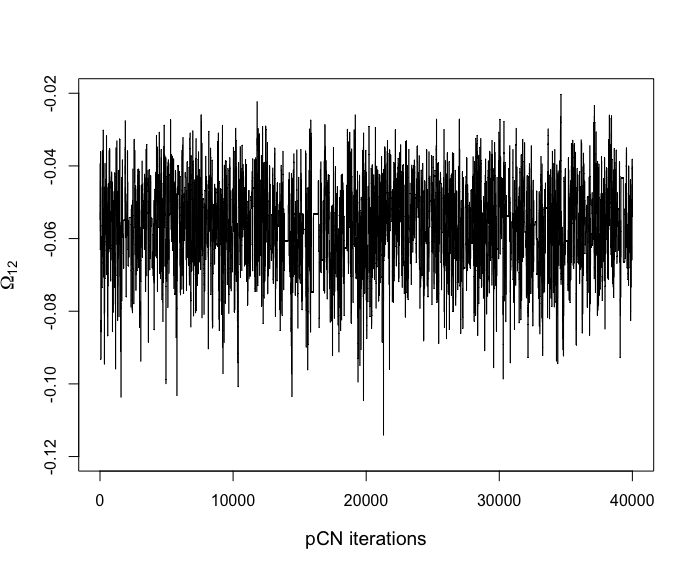}
\vspace{-0.5cm}
\caption{MpCN traceplots (first 3 plots) and pCN traceplots (last plot, bottom right) for the posteriors of $\Omega_{11}$, $\Omega_{22}$, $\Omega_{12}$ given the data in Fig.~\ref{fig:app1}.}
\label{fig:app3}
\end{center}
\vspace{-0.2cm}
\end{figure}

This design gives rise to a flexible, intuitive model, 
building upon mean-reverting-type dynamics imposed directly on the covariance matrix $\Sigma_t$. Such direction differs from typical approaches 
in the literature that work with various covariance decompositions \citep[e.g.][]{dell:15}, thus lacking direct interpretation of the dynamical behaviour for $\Sigma_t$ itself. 
To simplify the inferential setting, we consider a case 
where: i) the unknown parameter is $\Omega$, and $\Omega\in P^+(q)$;
ii) the L\'evy process corresponds to a compound Poisson one, 
with 
iid jumps comprised of scalar, independent Exponential variables, with all involved parameters assumed known. 
Data $\{y_{i}\}$ are obtained at times $0<t_1<\cdots <t_n$, $n\ge 1$, with
 $[\,y_i\,|\,\Sigma_{t_i}\,]\sim N(0,\Sigma_{t_i})$, $1\le i \le n$.

We use data corresponding to $n=180$ daily log-returns of the Euro-Dollar and Sterling-Dollar exchange rates, as observed in the period from 23/06/20 to 01/03/21. Thus, we look at a bivariate scenario, $q=2$, with $t_i=i$.
The data -- after linearly detrending and rescaling by $\sqrt{1000}$ -- are shown in Fig.~\ref{fig:app1}. In contrast to the complex methodology used in \cite{dell:15} for performing Bayesian inference in the scalar case, we follow the direction of a pseudo-marginal algorithm \citep[e.g.][]{MR2502648, MR2758115} to 
treat the latent paths of the covariance process $\{\Sigma_t\}$, thereby
replacing the intractable likelihood with unbiased estimators obtained by a particle filter. For the driving compound Poisson process, we use intensity $0.4$ and Exponential jumps with mean equal to $1/60$. Such values were selected via preliminary runs of MCMC algorithms with these parameters assumed unknown. We fix $\sigma_0=\mathrm{diag}\{0.05,0.05\}$, based again on preliminary runs. The prior for $\Omega$ is the inverse Wishart distribution 
$W^{-1}_2(2,I_2)$,
we stress that pursuing joint estimation of $\Omega$ with some of the above parameters in our bivariate setting is quite a challenging task, that is beyond the context of our work, and one that would deviate from our main interest of investigating the performance of algorithms we have defined on the previous parts of the paper and treat posterior sampling on $P^{+}(q)$.   


We applied MpCN and pCN, with $p=2$, both on the prior and the posterior distribution, on an Intel Xeon E5-2690, 2.9GHz, Memory 80GB (with $\texttt{C}$ code). 
Our theory has shown than pCN is not geometrically ergodic when applied on the Inverse-Wishart prior, and this is manifested in the contrasting behaviour of the MpCN and pCN traceplots in Fig.~\ref{fig:app2} -- for both algorithms we use $\rho=0.9$ giving an average acceptance probability of 46\% and 48\%, for MpCN and pCN respectively. When targeting the posterior, 
for the particle filter we use $1,000$ particles and dynamic resampling 
with ESS threshold $25\%$ (we made no use of parallelisation and applied the vanilla bootstrap filter) -- see \cite{MR2758115} for details on the use of a particle filter as means of obtaining an unbiased estimate of the likelihood of data to be used within the pseudo-marginal MCMC.
%
Fig.~\ref{fig:app3} shows MpCN and pCN traceplots generated using $\rho=0.999$, 
giving an average acceptance probability of 22\% and 14\%, for MpCN and pCN respectively.
Both algorithms required approximately 1 hour per 10,000 iterations. 
\revision{In this case, the consideration of $n=180$ data points under a Gaussian likelihood has flattened the tails of the posterior, so MpCN and pCN have similar performance when applied on the posterior given all $n=180$ observations. Note that the likelihood term of even a single data point in this example will change the mathematical nature of the decay of the tails of the target.}

\revision{
A key message we aim to convey is the \emph{robustness} of MpCN in all settings considered in this application vs the poor performance of pCN when applied on the prior distribution.} 

%

\section{Conclusions and Future Work}
\label{sec:conclude}
Our work presents one of the first contributions towards a systematic analysis 
(including derivation and ergodicity properties) of MCMC algorithms on matrix spaces. A number of interesting directions have now opened up to future research. We summarise some of them here.
\begin{itemize}
\item[(i)] 
As a starting point, we focused on blind proposals 
(this matched the requirements of the SV model). It is natural to move to the study of derivative-driven methods, e.g.~MALA or HMC. 
\item[(ii)]
There are some last steps remaining to obtain of complete proof of geometric ergodicity of the newly developed MpCN algorithm  -- this is left for future research, building upon the progress made here.
\item[(iii)]
The SV model sets up a new research direction in financial statistics. It leads to further investigations on differential models on matrix spaces and Monte Carlo methods that can be effective therein. 
This is a challenging task, requiring calibration of  high-dimensional matrix-parameters for latent dynamical models (much beyond the $2\times2$ space, and 180 data points used here).
\item[(iv)] There is a need for development and analysis of MCMC methods on spaces of symmetric positive-definite matrices restricted in sub-domains of $P^{+}(q)$ arising in the field of  Gaussian graphical models. Here interest lies in exploring, e.g., the space of precision matrices, given zeros for a number of partial auto-correlations determined by graphs. For relevant references, see Section 3 of \cite{lenk:13} or \cite{wang:12}. 
Indicatively, application of the upcasting approach in this context will be important, as it will open up directions for the development of effective MCMC algorithms that respect the space restrictions. 
Scalability with respect to dimensionality $q$ is also of high significance in this setting.
\end{itemize}

\section*{Acknowledgements} 
We thank an anonymous referee and the Associate Editor for suggestions that have greatly improved the contents of this paper.





\appendix

\section{Proof of Proposition \ref{prop:strange}}
\label{app:strange}

First we prove (i). 
From Proposition 5.1 of \cite{RT} and continuity of 
$x\mapsto P(x,\{x\})$, it suffices to show that
$\sup_{x\in M(p,q)}P(x,\{x\})=1$.  By definition, 
    \begin{equation}
    \label{eq:gareth}
    1-P(x,\{x\})=\int_{M(p,q)}\alpha(x,y(w))\phi_{p,q}(w;0,U,V)\dif w, 
    \end{equation}
    for $y=y(w)=\rho^{1/2}x+(1-\rho)^{1/2}w$; 
    notice also that the acceptance probability writes as,
    $$ 
    \alpha(x,y(w)) = \min\Big\{1,  \frac{\pi(y(w))\phi_{p,q}(x;0,U,V)}{\pi(x)\phi_{p,q}(y(w);0,U,V)}\Big\}.
    $$
    Let $x=rs$, for $r>0$ and $\|s\|_F=1$, so that $\|x\|_F=r\|s\|_F=r$. 
    For any $w\in M(p,q)$,
    \begin{align*}
    \|\tfrac{x}{r}-\tfrac{y(w)}{r}\|_F&=\|s-\rho^{1/2}s+r^{-1}~(1-\rho)^{1/2}w\|_F\\
    &\longrightarrow_{r\rightarrow\infty}1-\rho^{1/2}. 
    \end{align*}
    For each $w\in M (p,q)$ choose $r_0=r_0(w)>0$ so that 
    $\|\frac{x}{r}-\frac{y(w)}{r}\|_F<\varepsilon$ for any $r>r_0$. 
    Then, the definition
    of $C_{r,\varepsilon}(s)$ in the statement of the Proposition implies that for $r=\|x\|_F>r_0$,
    \begin{align}
    \label{eq:11}
    |\,\log\pi(y(w))-\log\pi(x)\,|= 
    |\,\log\pi(r\tfrac{y(w)}{r})-\log\pi(r\tfrac{x}{r})\,|\le C_{r,\varepsilon}(r^{-1}x)~r^2.
    \end{align}
    We have that, for fixed $w$, $y(w)=\rho^{1/2}x+\mathcal{O}(1)$ and, 
    \begin{align}
    \nonumber
            &-\log\phi_{p,q}(y(w);0, U, V)+\log\phi_{p,q}(x;0, U,V)\\[0.2cm]
             \nonumber
            &\qquad =\tfrac{1}{2}
            \big\{\,\tr\,[\,V^{-1}y(w)^{\top}U^{-1}y(w)\,]-\tr\,[\,V^{-1}x^{\top}U^{-1}x\,]\,\,\big\}
            \\[0.2cm]
             \nonumber
            &\qquad=-\tfrac{1-\rho}{2}\tr\,[\,V^{-1}x^{\top}U^{-1}x\,]+
            \mathcal{O}(\|x\|_F)\\[0.2cm]
            &\qquad\le -\tfrac{1-\rho}{2}\lambda^{-2}\|x\|^2_F+\mathcal{O}(\|x\|_F).  \label{eq:12}
    \end{align}
    Choose a sequence $x_n\in M(p,q)$, for $n=1,2,\ldots$, with $r_n=\|x_n\|_F \ge r_0$, such that, 
$$
\limsup_{n\rightarrow\infty}C_{r_n,\varepsilon}(s_n)<\tfrac{1-\rho}{2}\lambda^{-2},
$$
where $s_n=x_n/\|x_n\|_F$. Then, for $y_n(w)=\rho^{1/2}x_n+(1-\rho)^{1/2}w$, we have, 
\begin{align*}
\alpha(x_n,y_n(w))&=\min\Big\{1, \frac{\pi(y_n(w))\phi_{p,q}(x_n;0,U,V)}{\pi(x_n)\phi_{p,q}(y_n(w);0,U,V)}\Big\}\\[0.1cm]
&\le 
\frac{\pi(y_n(w))\phi_{p,q}(x_n;0,U,V)}{\pi(x_n)\phi_{p,q}(y_n(w);0,U,V)}\\[0.1cm]
&= 
\exp\left(\log\pi(y_n(w))-\log\pi(x_n)\right) \\ & \qquad \qquad \qquad \times \exp\left(-\log\phi_{p,q}(y_n(w);0,U,V)+\log\phi_{p,q}(x_n;0,U,V)\right)\\[0.1cm]
&\le \exp\big(\big(C_{r_n,\varepsilon}(s_n)-\tfrac{1-\rho}{2}\lambda^{-2}\big)\|x_n\|^2_F+\mathcal{O}(\|x_n\|_F)\big)\\
& \qquad \longrightarrow_{n\rightarrow\infty} 0,
\end{align*}
for any fixed $w\in M(p,q)$. The proof of (i) is completed via the dominated convergence theorem, 
as from (\ref{eq:gareth}) we obtain that $\sup P(x,\{x\})=1$. 

The proof for (ii) is similar. 
Let $x=rs$, for $r>0$ and $\|\,\U ^{-1/2}s\,\|_F=1$ as above. Then, for any $w\in M(p,q)$, 
\begin{align*}
   & \mathsf{d}\big( x^{\top}U^{-1}x,\, 
    y(w)^{\top}U^{-1}y(w)\big)
    =
    \mathsf{d}\big(\tfrac{x}{r}^{\top}\U ^{-1}\tfrac{x}{r},\, 
    \tfrac{y(w)}{r}^{\top}\U ^{-1}\tfrac{y(w)}{r}\big)\\
    &=
    \mathsf{d}\big(s^{\top}\U ^{-1}s,\, 
    (s-\rho^{1/2}s+r^{-1}(1-\rho)^{1/2}w)^{\top}\U ^{-1}(s-\rho^{1/2}s+r^{-1}(1-\rho)^{1/2}w)\big)\\
    &\rightarrow 
    \mathsf{d}\big(s^{\top}\U ^{-1}s, \,
    (s-\rho^{1/2}s)^{\top}\U ^{-1}(s-\rho^{1/2}s)\big)\\
    &=\mathsf{d}\big(I_q,\,(1-\rho^{1/2})^2I_q\big)
    =q^{1/2}2\log(1-\rho^{1/2}). 
\end{align*}
Let $\pi(x)=\tilde{\pi}(x^{\top}\U ^{-1}x)$. 
For each $w\in M(p,q)$, there exists $z_0=z_0(w)>0$ such that, if 
$\|\,\U ^{-1/2}x\,\|_F=z>z_0$ then,
\begin{align*}
    |\,\log\pi(y(w))-\log\pi(x)\,|\le C_{z,\varepsilon}(z^{-2}x^{\top}\U ^{-1}x)~z^2. 
\end{align*}
Observe that the trace of $z^{-2}x^{\top}\U ^{-1}x$ is $1$. 
As in the proof of (i), 
\begin{align*}
-\log\phi_{p,q}(y(w);0, U, V)+\log\phi_{p,q}(x;0, U,V)&=-\tfrac{1-\rho}{2}\tr\,[\,V^{-1}x^{\top}U^{-1}x\,]+
            \mathcal{O}(\|\U ^{-1/2}x\|_F)\\[0.2cm]
            &\le -\tfrac{1-\rho}{2}\tilde{\lambda}^{-2}\|\U ^{-1/2}x\|^2_F+\mathcal{O}(\|\U ^{-1/2}x\|_F).  
 \end{align*}
 The rest of the proof is the same as above. 

\section{Proof of Proposition \ref{prop:drift}}
\label{app:drift}

Recall the concept of rapid variation in Definition \ref{def:RAV}.
%
%
%
In the next lemma, we show a key property of probability measures on $P^+(q)$ with rapidly varying densities. 
\begin{lem}\label{lem:probability_fraction}
Suppose that the target $\tilde{\Pi}(\dif S)$ on $P^+(q)$ has a strictly positive, continuous, rapidly varying density $\tilde{\pi}(S)$ with respect to $\mu(\dif S)$. 
Then, w.p.1,
\[
\lim_{\tr(S)\rightarrow+\infty}\Big|\log\Big(\frac{\tilde{\pi}(\epsilon\circ S)}{\tilde{\pi}(S)}\Big)\Big|=+\infty,
\]
where $\epsilon\sim\mathcal{L}(\dif\epsilon)$ -- see Proposition \ref{prop:random_walk} for the definition of $\mathcal{L}(\dif \epsilon)$.  
\end{lem}

\begin{proof}
By tightness of $\mathcal{L}(\dif\epsilon)$, for any $\delta\in (0,1)$, there exists $r\in (0,1)$ such that,
\[
\mathbb{P}\,[\,r\le \lambda_{\min}(\epsilon)\le \lambda_{\max}(\epsilon)\le r^{-1}\,]> 1-\delta,
\]
where $\lambda_{\min}(S)$ and $\lambda_{\max}(S)$ are the smallest and largest eigenvalues of $S\in P^+(q)$. 
Thus, for any $C>0$, $\delta\in (0,1)$, we have,
\[
\mathbb{P}\,\Big[\,\Big|\log\Big(\frac{\tilde{\pi}(\epsilon\circ S)}{\tilde{\pi}(S)}\Big)\Big|\le C\,\Big]\le \delta + 
\mathbb{P}\,[\,\epsilon\in A(S)\,],
\]
where, 
\[
A(S):=\Big\{\epsilon\in P^+(q): \Big|\log\Big(\frac{\tilde{\pi}(\epsilon\circ S)}{\tilde{\pi}(S)}\Big)\Big|\le C, r\le \lambda_{\min}(\epsilon)\le \lambda_{\max}(\epsilon)\le r^{-1}\Big\}. 
\]
We write  
$\operatorname{Leb}^{q(q+1)/2}$ for $\prod_{1\le i\le j\le q}\dif S_{ij}$. 
Recall that reference measure $\mu(\dif S)$ has density $(\det(S))^{-(q+1)/2}$ with respect to $\operatorname{Leb}^{q(q+1)/2}$. 
Thus,  $\mathcal{L}(\dif\epsilon)$ has bounded density (with respect to $\operatorname{Leb}^{q(q+1)/2}$) on, $$\{\epsilon\in P^+(q):\,r\le \lambda_{\min}(\epsilon)\le \lambda_{\max}(\epsilon)\le r^{-1}\}.$$ 
So, there exists a constant $c>0$ such that, 
\[
\mathbb{P}\,[\,\epsilon\in A(S)\,]\le 
c\operatorname{Leb}^{q(q+1)/2}[\,\epsilon\in A(S)\,]. 
\]
Now, we consider the variable transformation $\epsilon=\mathsf{U}^\top\Lambda\mathsf{U}$ where $\mathsf{U}\in\mathcal{O}(q)$ and $\Lambda$ is a diagonal matrix with positive diagonal elements $r\le \lambda_1<\lambda_2<\cdots<\lambda_q\le r^{-1}$. From Theorem 5.3.1 of \cite{MR770934}, for some constant $c>0$, we have,
\begin{align*}
\operatorname{Leb}^{q(q+1)/2}[\,\epsilon\in A(S)\,]
&=c
\int_{\mathsf{U}\in\mathcal{O}(q)}\dif \mathsf{U}~\int 1_{\{\mathsf{U}^\top\Lambda\mathsf{U}~\in A(S)\}}\prod_{i<j}(\lambda_i-\lambda_j)\dif \lambda_1\cdots\dif \lambda_q\\
&=
c\int_{\mathsf{U}\in\mathcal{O}(q)}\dif \mathsf{U}~\int 1_{\{\mathsf{U}^\top\Lambda(\lambda_1,\delta)\mathsf{U}~\in A(S)\}}\prod_{i<j}(\delta_i-\delta_j)\dif \lambda_1\dif \delta_2\cdots\dif \delta_q
\end{align*}
where $\delta_1=0$, $\delta=\{\delta_i=\lambda_i-\lambda_1;\, i=2,\ldots,q\}$ and $\Lambda(\lambda_1,\delta)$ is diagonal matrix with values $\lambda_1, \lambda_1+\delta_2,\ldots, \lambda_1+\delta_q$. Then,
\begin{align*}
\operatorname{Leb}^{q(q+1)/2}&[\,\epsilon\in A(S)\,] \\ &\le c~r^{-q(q-1)/2}
\int_{\mathsf{U}\in\mathcal{O}(q)}\int_{0< \delta_2<\cdots<\delta_q<r^{-1}}\operatorname{Leb}\,[\,A_r(S,\mathsf{U},\delta)\,]\,\dif \mathsf{U}\,\dif \delta_2\cdots\dif\delta_q,
\end{align*}
where,
\[
A_r(S,\mathsf{U},\delta)=\{\,r<\lambda_1<r^{-1}-\delta_q: \mathsf{U}^{\top}\Lambda(\lambda_1,\delta)\mathsf{U}\in A(S)\,\}\subset\mathbb{R}. 
\]
\revision{The domain of the integral is always contained in a compact set $\mathcal{O}(q)\times [0, r^{-1}]^{d-1}$. Moreover, $A_r(S, U,\delta)$ is a subset of an open interval $(r, r^{-1})$. Thus, we can use the dominated convergence theorem for $\tr(S)\rightarrow \infty$ for the left-hand side of the above inequality. }
By the dominated convergence theorem, it suffices to show $\operatorname{Leb}(A_r(S,\mathsf{U},\delta))\longrightarrow_{\operatorname{tr}(S)\uparrow+\infty}~ 0$ for each $r>0$, $\mathsf{U}\in\mathcal{O}(q)$ and $\delta$. 
Observe that,
\begin{align*}
	A_r(S,\mathsf{U},\delta)&\otimes A_r(S,\mathsf{U},\delta)\\ &\subset \Big\{(\lambda_1,\lambda_1'): r\le \lambda_1, \lambda_1'\le r^{-1}-\delta_q,\  \Big|\log\Big(\frac{\tilde{\pi}((\mathsf{U}^\top\Lambda(\lambda_1,\delta)\mathsf{U})\circ S)}{\tilde{\pi}((\mathsf{U}^\top\Lambda(\lambda_1',\delta)\mathsf{U})\circ S)}\Big)\Big|\le 2C\Big\}. 
\end{align*}
Since $\pi(x)=\tilde{\pi}(x^\top U^{-1}x)$ is a rapidly varying function, we have,
\[
\lambda_1\neq \lambda_1'~\Longrightarrow~\Big|\log\Big(\frac{\tilde{\pi}((\mathsf{U}^\top\Lambda(\lambda_1,\delta)\mathsf{U})\circ S)}{\tilde{\pi}((\mathsf{U}^\top\Lambda(\lambda_1',\delta)\mathsf{U})\circ S)}\Big)\Big|\longrightarrow_{\tr (S)\uparrow\infty}+\infty;
\]
we used the fact that if $a>b$ then $\mathsf{U}^\top\Lambda(a,\delta)\mathsf{U}-\mathsf{U}^\top\Lambda(b,\delta)\mathsf{U}\in P^+(q)$. The probability of $\lambda_1=\lambda_1'$ is $0$, so we can complete the claim by the dominated convergence theorem. 
\end{proof}

Let $\tilde{P}_{\rho}$ be the MpCN kernel on $P^{+}(q)$ -- see the statement of Proposition \ref{prop:drift}.

\begin{proof}[\textbf{Proof or Proposition \ref{prop:drift}}]
%
The MpCN proposal in $P^+(q)$ can be written as $\epsilon\,\circ\,S$, with $\epsilon\sim\mathcal{L}(\dif\epsilon)$, so we have,
for $\mathsf{V}(S)=\tilde{\pi}(S)^{-\alpha}$, $\alpha\in(0,1)$,
\begin{align*}
\frac{\tilde{P}_{\rho}\mathsf{V}(S)-\mathsf{V}(S)}{\mathsf{V}(S)}
&=\mathbb{E}\,\Big[\,\Big\{\Big(\frac{\tilde{\pi}(\epsilon\circ S)}{\tilde{\pi}(S)}\Big)^{-\alpha}-1\Big\}\min\Big\{1,\frac{\tilde{\pi}(\epsilon\circ S)}{\tilde{\pi}(S)}\Big\}\,\Big]\\
&=\mathbb{E}\,\Big[\,\Big(\frac{\tilde{\pi}(\epsilon\circ S)}{\tilde{\pi}(S)}\Big)^{-\alpha}\min\Big\{1,\frac{\tilde{\pi}(\epsilon\circ S)}{\tilde{\pi}(S)}\Big\}\,\Big]
-
\mathbb{E}\,\Big[\,\min\Big\{1,\frac{\tilde{\pi}(\epsilon\circ S)}{\tilde{\pi}(S)}\Big\}\,\Big]. 
\end{align*}
By dominated convergence theorem and Lemma \ref{lem:probability_fraction}, the first term in the right-hand side converges to $0$ as $\tr(S)\rightarrow\infty$. Also, by Lemma \ref{lem:probability_fraction}, the limit of the second term is,
\[
\liminf_{\tr(S)\rightarrow+\infty}\mathbb{E}\,\Big[\,\min\Big\{1,\frac{\tilde{\pi}(\epsilon\circ S)}{\tilde{\pi}(S)}\Big\}\,\Big]=
\liminf_{\tr(S)\rightarrow+\infty}\mathbb{P}\,\Big[\,\log\Big(\frac{\tilde{\pi}(\epsilon\circ S)}{\tilde{\pi}(S)}\Big)>0\,\Big]. 
\]
The proof will be completed if we show that the right-hand side is strictly greater than $0$. 
Thanks to the rapidly varying property, 
\begin{align*}
\liminf_{\tr(S)\rightarrow+\infty}\mathbb{P}\,\Big[\,\log\Big(\frac{\tilde{\pi}(\epsilon\circ S)}{\tilde{\pi}(S)}\Big)>0\,\Big]&\ge 
\liminf_{\tr(S)\rightarrow+\infty}\mathbb{P}\,\Big[\,\log\Big(\frac{\tilde{\pi}(\epsilon\circ S)}{\tilde{\pi}(S)}\Big)>0, \,I-\epsilon\in P^+(q)\,\Big]\\
&= 
\mathbb{P}\,[\, I-\epsilon\in P^+(q)\,]. 
\end{align*}
Since the law of $\mathcal{L}(\dif\epsilon)$ is absolutely continuous with respect to $\mu$, and $\mu\,[\,\{\epsilon;I-\epsilon\in P^+(q)\}\,]$ is positive,  the probability is positive. Thus, the drift inequality for $\tr (S)\rightarrow\infty$ follows. 
\end{proof}

\section{Proof of Proposition \ref{prop:spetral_gap_quivalence}}
\label{app:rho}


Consider the standard Hilbert space 
$L^2(\Pi)=\{f:M(p,q)\rightarrow\mathbb{R}; \|f\|^2_{L^2(\Pi)}<\infty\}$
where $\|f\|_{L^2(\Pi)}=\sqrt{\langle f,f\rangle}$ with inner product,
$$
\langle f,g\rangle = \langle f, g\rangle_{L^2(\Pi)}=\int_{x\in M(p,q)} f(x)g(x)\Pi(\dif x). 
$$
A Markov kernel $P$ on $M(p,q)$ is a linear operator on this space via the action,
$$
(Pf)(x)=\int_{y\in E}P(x,\dif y)f(y),\quad f\in L^2(\Pi).
$$
A linear operator $P$ is self-adjoint if 
$
\langle Pf, g\rangle=\langle f,Pg\rangle$, $f, g\in L^2(\Pi)$;
this is equivalent to $\Pi$-reversibility of $P$. Also, a linear operator is positive if 
$
\langle f, Pf\rangle \ge 0$,  $f\in L^2(\Pi)$.
By Section XI.8 of \cite{MR1336382}, if $P$ is self-adjoint, its spectrum lies on the real line; if it is also a  positive operator, its spectrum lies on $\mathbb{R}_{+}$. 
First we will show positivity of the pCN kernel. Using this, we will show positivity of the MpCN kernel. 

\begin{lem}
The pCN kernel $P(x,\dif y)$ on $M(p,q)$ is a positive, self-adjoint linear operator on $L^2(\Pi)$. 
\end{lem}

\begin{proof}
First, we show that the proposal kernel $Q$ of pCN is positive, self-adjoint on $L^2(\Pi')$, where, 
\[
Q(x,\cdot)=N_{p,q}(\rho^{1/2}x, (1-\rho)U, V),\quad  \Pi'=N_{p,q}(0, U, V). 
\]
The self-adjointness property is immediate upon observing,
\[
\phi_{p,q}(y;\rho^{1/2}x, (1-\rho)U, V)~\phi_{p,q}(x; 0, U, V)=
\phi_{p,q}(x;\rho^{1/2}y, (1-\rho)U, V)~\phi_{p,q}(y; 0, U, V). 
\]
By the reproductive property of the Gaussian distribution (Remark \ref{rem:connect}(i, ii)),  we have
$$Q^{1/2}(x,\cdot)=N_{p,q}(\rho^{1/4}x, (1-\rho^{1/2})U, V),$$ where the linear operator $Q^{1/2}$ is defined by $Q^{1/2}(Q^{1/2}f)=Qf$, $f\in L^2(\Pi')$.   Observe that $Q^{1/2}$ is the same as $Q$ with $\rho^{1/2}$ in the place of $\rho$. By this fact,  $Q^{1/2}$ is  also self-adjoint in $L^2(\Pi')$. Then,
\begin{align*}
    \langle f, Qf\rangle_{L^2(\Pi')}=\langle f, Q^{1/2}Q^{1/2}f\rangle_{L^2(\Pi')}
    =\|Q^{1/2}f\|_{L^2(\Pi')}\ge 0. 
\end{align*}
%
We now show that $P$ itself is positive, self-adjoint on 
$L^2(\Pi)$. Since $P$ is a $\Pi$-reversible Metropolis kernel, we proceed to the proof of positiveness. We use the decomposition approach introduced in Lemma 3.1 of \cite{MR3078012}. 
Recall that the acceptance probability is 
$
\alpha(x,y)=\min\{1,\tfrac{\pi(y)}{\pi(x)}\},
$
where $\pi(x)$ is the density of $\Pi$ with respect to $\Pi'$. We can write, 
\[
\pi(x)~\alpha(x,y)=\min\left\{\pi(x), \pi(y)\right\}=\int_0^\infty 1_{K(t)}(x)1_{K(t)}(y)\dif t
\]
where 
$
K(t)=\{x\in M(p,q): \pi(x)\ge t\} 
$.
Using the above, we have,
\begin{align*}
    \langle f, Pf\rangle_{L^2(\Pi)}&=\int_{x,y\in M(p,q)} f(x)f(y)\Pi(\dif x)P(x,\dif y)\\
    &\ge \int_{x,y\in M(p,q)} f(x)f(y)\alpha(x,y)\Pi(\dif x)Q(x,\dif y)\\
    &= \int_{x,y\in M(p,q)} f(x)f(y)\alpha(x,y)~\pi(x)\Pi'(\dif x)~Q(x,\dif y)\\
    &= \int_0^\infty\int_{x,y\in M(p,q)} [1_{K(t)}f](x)~[1_{K(t)}f](y)\Pi(\dif x)Q(x,\dif y)\dif t\\
    &= \int_0^\infty\langle 1_{K(t)}f, Q[1_{K(t)}f]\rangle_{L^2(\Pi')}~\dif t\ge 0. 
\end{align*}
We used the fact that $1_{K(t)}f\in L^2(\Pi')$ by Markov's inequality. Thus, we have that $P$ is a positive operator. 
\end{proof}

Next, we prove positivity and self-adjointness of the MpCN kernel. Such properties are inherited by the pCN kernel.
Let,
\begin{equation}\label{eq:joint_distribution}
Q_{V}(x,\cdot )=N_{p,q}(\rho^{1/2}x, U, (1-\rho)V),\quad \Pi'_V=N_{p,q}(0, U, V),
\end{equation}
be the proposal kernel of pCN and its invariant distribution, respectively, for matrix parameter $V$. 
Observe that
\begin{align*}
    \nu_U(\mathrm{d}x)W_q^{-1}(\dif V;p, x^{\top}U^{-1}x)
    &=
    \Pi_V'(\dif x)\mu(\dif V).
\end{align*}
As in the proof of Lemma \ref{lem:kernel_mpcn} and equation (\ref{eq:joint}), following the Bayesian paradigm construction of the MpCN kernel, we have,
\begin{align*}
\nu_U(\dif x)Q(x, \dif y)
&=\int_{V\in P^+(q)}\nu_U(\dif x)W_q^{-1}(\dif V;p, x^{\top}U^{-1}x)Q_V(x, \dif y)\\
&=
\int_{V\in P^+(q)}\Pi_V'(\dif x)Q_V(x, \dif y)\mu(\dif V),
\end{align*}
where $Q$ is the MpCN kernel. 

\begin{prop}
The MpCN kernel $P(x,\dif y)$ is a positive, self-adjoint linear operator on $L^2(\Pi)$. 
\end{prop}

\begin{proof}
Self-adjointness follows from reversibility. We will prove the positivity part. 
As in the pCN kernel case, for the target probability distribution $\Pi(\dif x)=\pi(x)\nu_U(\dif x)$ we have,
\begin{align*}
    \langle f, Pf\rangle_{L^2(\Pi)}&=\int f(x)f(y)\Pi(\dif x)P(x,\dif y)\\
    &\ge \int_{x,y\in M(p,q)} f(x)f(y)\alpha(x,y)\Pi(\dif x)Q(x,\dif y)\\
    &= \int_0^\infty\int_{x,y\in M(p,q)} [1_{K(t)}f](x)[1_{K(t)}f](y)\nu_U(\dif x)Q(x,\dif y)\dif t. 
\end{align*}
By (\ref{eq:joint_distribution}), together with positivity and self-adjointness of the pCN proposal kernel, we obtain,
\begin{align*}
    \langle f, Pf\rangle_{L^2(\Pi)}
    &\ge c_{p,q}^*~\int_0^\infty\int_{V\in P^+(q)}\int_{x,y\in M(p,q)} [1_{K(t)}f](x)[1_{K(t)}f](y)\Pi_V'(\dif x)Q_V(x,\dif y)\mu(\dif V)\dif t\\
    &= c_{p,q}^*~\int_0^\infty\int_{V\in P^+(q)}\langle 1_{K(t)}f, Q_V[1_{K(t)}f]\rangle_{L^2(\Pi_V')}\mu(\dif V)\dif t\ge 0. 
\end{align*}
Here, we used the fact that $1_{K(t)}f\in L^2(\Pi_V')$ by Markov's inequality.  
\end{proof}

\noindent Thus, all eigenvalues of the MpCN kernel are 
in $\mathbb{R}_{+}$. Let $\mathrm{spec}(P)$ be the  eigenvalues of a self-adjoint Markov kernel $P$ on $L^2(\Pi)$ -- excluding the constant function.  A transition kernel $P$ is said to \emph{have a spectral gap} if $1-|\sup\mathrm{spec}(P)|>0$ -- the latter being equivalent to $1-\sup\mathrm{spec}(P)>0$ when $P$ is a positive operator.  It is known that, if $P$ is positive and self-adjoint, it is geometrically ergodic if and only if it has a spectral gap  \citep[see, e.g.,][]{MR1448322, MR1915532}. 
Note that, by Theorem XI.8.2 of \cite{MR1336382}, the spectral gap $1-\sup\mathrm{spec}(P)$ writes as, 
\begin{equation}
\label{eq:gap}
\inf_{f\in L^2(\Pi);\, \Pi(f)=0}\frac{\mathcal{E}_P(f,f)}{\|f\|^2}=1-\sup_{f\in L^2(\Pi);\,\Pi(f)=0}\frac{\langle f,Pf\rangle}{\|f\|^2},
\end{equation}
where,
\[
\mathcal{E}_P(f,f)=\tfrac{1}{2}\int_{x,y\in M(p,q)} (f(x)-f(y))^2\Pi(\dif x)P(x,\dif y). 
\]
We study the spectral gap of the MpCN kernel. 
Recall from the main text, that to stress the involvement of parameter $\rho$, we write the MpCN kernel as $P_\rho$. We denote by $Q_\rho$ the proposal kernel for $P_\rho$. 


\begin{proof}[\textbf{Proof of Proposition \ref{prop:spetral_gap_quivalence}}]
Part (i): \\ 
First we note that,
\[
(1-\rho)x^\top U^{-1}x+(y-\rho^{1/2}x)^\top U^{-1}(y-\rho^{1/2}x)\le 2x^\top U^{-1}x+2y^\top U^{-1}y,
\]
where the left-hand side is equal to $R(x,y)$ defined in Lemma \ref{lem:kernel_mpcn}
; 
the above inequality follows from the right-hand side being equal to, 
\[
[(1-\rho)x^\top U^{-1}x+(y-\rho^{1/2}x)^\top U^{-1}(y-\rho^{1/2}x)]+
[(1-\rho)x^\top U^{-1}x+(y+\rho^{1/2}x)^\top U^{-1}(y+\rho^{1/2}x)]. 
\]
Let $\mathsf{q}_\rho(x,y)$ be the probability density function of $Q_\rho(x,\cdot)$ with respect to $\nu$. Via the explicit form of $\mathsf{q}_\rho(x,y)$ (see Lemma \ref{lem:kernel_mpcn}
), the above inequality implies,
\[
\mathsf{q}_\rho(x,y)\ge 2^{-pq}\mathsf{q}_0(x,y). 
\]
Using this inequality, 
\begin{align*}
\mathcal{E}_{P_\rho}(f,f)&=\frac{1}{2}\int_{x,y\in M(p,q)}(f(x)-f(y))^2\Pi(\dif x)P_\rho(x,\dif y)\\
&=\frac{1}{2}\int_{x,y\in M(p,q)}(f(x)-f(y))^2\Pi(\dif x)\alpha(x,y)Q_\rho(x,\dif y)\\
&\ge 2^{-pq}\frac{1}{2}\int_{x,y\in M(p,q)}(f(x)-f(y))^2\Pi(\dif x)\alpha(x,y)Q_0(x,\dif y)\\
&\ge 2^{-pq}\mathcal{E}_{P_0}(f,f). 
\end{align*}
Thus, existence of a spectral gap of $P_0$ implies that $P_\rho$ also has one. \\
Part (ii):\\
We define the set, 
\begin{equation*}
L^2_\star(\Pi):=\{f:M(p,q)\rightarrow\mathbb{R}:\,f(x)=g(x^{\top}U^{-1}x)\,\,\,\textrm{for a}\,\,g\in L^{2}(\tilde{\Pi}) \} \subseteq L^2(\Pi).
\end{equation*}
Spaces  $L^{2}_{\ast}(\Pi)$, $L^2(\tilde{\Pi})$ are isomorphic via the mapping $f\leftrightarrow g$.
For $f\in L^2_{\star}(\Pi)$ we have that 
$Pf(x)\in  L^2_{\star}(\Pi)$, thus operator $P$ restricted on $L^2_{\star}(\Pi)$ is positive, self-adjoint, with
$Pf(x) \equiv \tilde{P}g(s)$, where $s=x^{\top}U^{-1}x$.
It follows trivially, that we also have 
$\mathcal{E}_{P_\rho}(f,f)\ge 2^{-pq}\mathcal{E}_{P_0}(f,f)$ -- 
as obtained in Part (i) -- for $f\in L^2_\star(\Pi)$,
or equivalently $\mathcal{E}_{\tilde{P}_\rho}(g,g)\ge 2^{-pq}\mathcal{E}_{\tilde{P}_0}(g,g)$, for  $\mathcal{E}_{\tilde{P}_\rho}(g,g)$ defined in an obvious way.
The proof is now complete.
%
%
%
\end{proof}

\bibliographystyle{Chicago}

\end{document}